%% file: main.tex
\documentclass[a4paper]{llncs}

\usepackage{amsmath,mathtools,amssymb,mathrsfs}
\usepackage{color}
\usepackage{tikz}
\usepackage{etoolbox}
\usepackage{txfonts}
\usepackage{multibib}
\usepackage{url}

\newcites{addr}{Additional References}

\usepgflibrary{shapes}
\usetikzlibrary{arrows,shapes,positioning,snakes,automata,backgrounds,petri,calc}

\tikzstyle{max}=[thick,draw,minimum size=1.2em,inner sep=0em]
\tikzstyle{min}=[diamond,thick,draw,minimum size=1.4em,%
inner sep=0em]
\tikzstyle{ran}=[circle,thick,draw,minimum size=1.2em,%
inner sep=0em]
\tikzstyle{act}=[circle,thick,draw,fill,minimum size=.7em,%
inner sep=0em]
\tikzstyle{mc}=[rounded corners,thick,draw,minimum size=1.4em,%
inner sep=.5ex]
\tikzstyle{tran}=[thick,draw,->,>=stealth]
\tikzstyle{loop left}=[tran, to path={.. controls +(150:.5)
	and +(210:.5) .. (\tikztotarget) \tikztonodes}]
\tikzstyle{loop right}=[tran, to path={.. controls +(30:.5)
	and +(330:.5) .. (\tikztotarget) \tikztonodes}]
\tikzstyle{loop above}=[tran, to path={.. controls +(60:.5)
	and +(120:.5) .. (\tikztotarget) \tikztonodes}]
\tikzstyle{loop below}=[tran, to path={.. controls +(240:.5)
	and +(300:.5) .. (\tikztotarget) \tikztonodes}]
\usepackage{todonotes}

\newcommand{\defineNote}[3][black!65!green]{\expandafter\def\csname #2\endcsname
##1{\stepcounter{fixcount}\fxwarning{\textcolor{#1}{\textbf{#3}: ##1}}}}

\newcommand{\states}{S}
\newcommand{\cstates}{S_{\Box}}
\newcommand{\stochstates}{S_{\raisebox{0.75 pt}{{\scalebox{0.5}{$\bigcirc$}}}}}
\newcommand{\trans}{T}
\newcommand{\rev}{r}
\newcommand{\runs}[1]{\run_{\M}(#1)}
\newcommand{\energy}{E}

\newcommand{\out}{\mathit{out}}
\newcommand{\enlev}[3]{\def\EmptyTest{#2}\ifdefempty{\EmptyTest}{\mathit{Lev}_{#1}^{(#3)}}{\mathit{Lev}_{#1}^{(#3)}(#2)}}
\newcommand{\MP}{\mathit{MP}}
\newcommand{\Probm}{\mathbb{P}}
\newcommand{\val}{\mathit{Val}}
\newcommand{\size}[1]{|\!|#1|\!|}
\newcommand{\maxE}{M_{\calE}}
\newcommand{\minsafe}{\mathit{min\text{-}safe}}
\newcommand{\minpump}{\mathit{min\text{-}pump}}
\newcommand{\sccp}{\textrm{SP-EMDP}}

\newcommand{\calF}{\mathcal{F}}
\newcommand{\calM}{\mathcal{M}}

\newcommand{\calT}{\mathcal{T}}

\newcommand{\E}{\mathbb{E}}
\newcommand{\calE}{\mathcal{E}}

\newcommand{\calL}{\mathcal{L}}

\newcommand{\M}{\mathcal{M}}

\newcommand{\Prob}{\mathit{Prob}}

\newcommand{\Nset}{\mathbb{N}}

\newcommand{\Zset}{\mathbb{Z}}
\newcommand{\Qset}{\mathbb{Q}}
\newcommand{\Rset}{\mathbb{R}}

\newcommand{\run}{\mathit{Run}}
\newcommand{\NP}{\mathsf{NP}}
\newcommand{\coNP}{\mathsf{coNP}}
\newcommand{\NPcoNP}{\mathsf{NP}\cap\mathsf{coNP}}
\newcommand{\PTIME}{\mathsf{P}}
\newcommand{\EG}{\mathsf{EG}}
\newcommand{\PSPACE}{\mathsf{PSPACE}}
\newcommand{\EXPTIME}{\mathsf{EXPTIME}}

\newcommand{\len}[1]{\mathit{len}(#1)}

\newcommand{\trend}{\mathit{trend}}

\renewcommand{\vec}[1]{\mathit{\pmb{#1}}}

\newcommand{\eps}{\varepsilon}

\newcommand{\cu}[1]{\mathbf{#1}}
\newcommand{\update}{\mathit{update}}
\newcommand{\nxt}{\mathit{next}}
\newcommand{\Lim}{\mathit{Lim}}

\newcommand{\test}[1]{\mathit{low}_{#1}}
\newcommand{\last}{\mathit{last}}
\newcommand{\pathstate}[2]{\mathit{St}(#1,#2)}

\newenvironment{reftheorem}[2]{\begin{trivlist}
\item[\hskip \labelsep {\bfseries #1}\hskip \labelsep {\bfseries #2}]\itshape}{\end{trivlist}}

\newcommand{\freq}[1]{f_{#1}}
\newcommand{\meanp}{\mathit{mp}}

\newcommand{\mart}[1]{m^{(#1)}}
\newcommand{\maralt}[1]{\hat{m}^{(#1)}}

\newcommand{\ttmart}[1]{\bar{m}^{(#1)}}

\title{Optimizing the Expected Mean Payoff in Energy Markov Decision Processes}

\author{Tom\'a\v{s} Br\'azdil\inst{1} \and 
	Anton\'{\i}n Ku\v{c}era\inst{1} \and
	Petr Novotn\'y\inst{2}\thanks{The research  has received funding from the 
	People Programme (Marie Curie Actions) of the European Union's Seventh 
	Framework Programme (FP7/2007-2013) under REA grant agreement no [291734].}}

\institute{%
	Faculty of Informatics MU, Botanick\'a 68a, 602\,00 Brno,
	Czech Republic, \email{\{brazdil,kucera\}@fi.muni.cz} \and
	IST Austria, Klosterneuburg, Austria, \email{petr.novotny@ist.ac.at}}

\begin{document}

\maketitle

\begin{abstract}
Energy Markov Decision Processes (EMDPs) are finite-state Markov decision processes where each transition is assigned an integer counter update and a rational payoff. An EMDP configuration is a pair $s(n)$, where $s$ is a control state and $n$ is the current counter value. The configurations are changed by performing transitions in the standard way. We consider the problem of computing a safe strategy (i.e., a strategy that keeps the counter non-negative) which maximizes the expected mean payoff. 
\end{abstract}

\input{intro}
\input{prelims}

\input{one-counter}

\input{results-tech-petr}
\input{thm2-proof}

\bibliographystyle{splncs03}
\bibliography{Disertace,disertace,tomas}
\appendix

\input{app-hardness}

\bibliographystyleaddr{splncs03}
\bibliographyaddr{Disertace,disertace,tomas}

\end{document}

%% file: intro.tex
\section{Introduction}
\label{sec-intro}

\emph{Resource-aware systems} are systems that consume/produce a discrete resource, such as (units of) time, energy, or money, along their runs. This resource is \emph{critical}, i.e., if it is fully exhausted along a run, a severe runtime error appears and such a situation should be avoided to the largest possible extent. Technically, resource-aware systems are modeled as finite-state programs operating over an integer counter representing the resource. A \emph{configuration} is a pair $s(n)$ where $s$ is the current control state and $n$ is the number of currently available resource units. Each transition is assigned an integer \emph{update} modeling the consumption/production of the resource caused by performing the transition. 
\medskip

\noindent
\textbf{Our Contribution.}
In this paper, we concentrate on the \emph{long-run average optimization problem} for resource-aware systems with both controllable and stochastic states. That is, we assume that the finite control of our resource-aware system is a finite-state Markov decision process (MDP), and each transition is assigned (in addition to the integer counter update) a rational \emph{payoff}\footnote{The payoff may correspond to some independent performance measure, or it can reflect the use of the critical resource represented by the counter.}. The resulting model is called \emph{energy Markov decision process (EMDP)}. Intuitively, given an EMDP and its initial configuration, the task is to compute a \emph{safe} strategy maximizing the \emph{expected mean payoff}.  Here, a strategy is safe if it ensures that the counter stays non-negative along all runs. The \emph{value} of a given configuration $s(n)$, denoted by $\val(s(n))$, is the supremum of all expected mean payoffs achievable by a safe strategy, and a strategy is \emph{optimal} for $s(n)$ if it is safe and achieves the value. Observe that $\val(s(n)) \geq \val(s(m))$ whenever $n \geq m$, and hence we can also define the \emph{limit value} of $s$, denoted by $\val(s)$, as $\lim_{n \rightarrow \infty} \val(s(n))$.

Since optimal safe strategies may not exists in general, the first natural question is the following:
\smallskip
 
\noindent
\textbf{[Q1]}. \emph{Can we determine a ``reasonable'' condition under which an optimal strategy exists?}
\smallskip
 
\noindent
By ``reasonable'' we mean that the condition should be decidable (with low complexity) and tight (i.e., we should provide counterexamples witnessing that optimal strategies do not necessarily exist if the condition is violated). Further, there are two basic algorithmic questions.
\smallskip
 
\noindent
\textbf{[Q2]}. \emph{Can we compute $\val(s(n))$ for a given configuration $s(n)$? If not, can we at least approximate the value up to a given absolute error $\varepsilon > 0$? Can we compute/approximate $\val(s)$ for a given state $s$? What is the complexity of these problems?}  
\smallskip 
 
\noindent
To show that computing an $\varepsilon$-approximation of $\val(s(n))$ is computationally hard, we consider the following \emph{gap threshold problem}: given a configuration $t(k)$ of a given EMDP and numbers $x,\eps$, where $\eps>0$, such that either $\val(t((k))\geq x$ or $\val(t(k))\leq x-\eps$, decide which of these two alternatives holds\footnote{Formally, the decision algorithm answers ``yes'' iff the first (or the second) possibility holds.}. Note that if the gap threshold problem is \mbox{$\textsf{X}$-hard} for some complexity class $\textsf{X}$, then $\val(s(n))$ cannot be \mbox{$\varepsilon$-approximated} in polynomial time unless $\textsf{X} = \PTIME$.

\noindent
\textbf{[Q3]}. \emph{Can we compute (a finite description of) an optimal strategy for a given configuration (if it exists)? For a given $\varepsilon > 0$, can we compute an \mbox{$\varepsilon$-optimal} strategy? How much memory is required by these strategies? What is the complexity of the strategy synthesis problems?}
\smallskip

Before formulating our answers to the above questions, we need to briefly discuss the relationship between EMDPs and \emph{energy games} \cite{CD:energy-parity-games,DBLP:conf/emsoft/ChakrabartiAHS03,DBLP:conf/formats/BouyerFLMS08}.
 
The problems of \textbf{[Q2]} and \textbf{[Q3]} subsume the question whether a given configuration of a given EMDP is safe. This problem can be solved by algorithms for 2-player non-stochastic energy games~\cite{CdAHS:resource-interfaces}, where we treat the stochastic vertices as if they were controlled by an adversarial player. The correctness of this approach stems from the fact that keeping the energy level non-negative is an objective whose violation is witnessed by a finite prefix of a run.  Let $\EG$ (\textsf{E}nergy \textsf{G}ames) be the problem of deciding whether a given configuration in a given energy game is safe. A \mbox{\emph{$\PTIME^{\EG}$~algorithm}} is a deterministic polynomial-time algorithm which inputs an EMDP $\calE$ (and possibly some initial configuration $s(n)$ of $\calE$) and uses an oracle which freely decides the safety problem for the configurations of $\calE$.  We assume that the counter updates and rewards used in $\calE$, and the $n$ in $s(n)$, are encoded as (fractions of) binary numbers. The size of $\calE$ and $s(n)$ is denoted by $\size{\calE}$ and $\size{s(n)}$, respectively. It is known that $\EG$ is solvable in pseudo-polynomial time, belongs to $\NP \cap \coNP$, and it is at least as hard as the parity game problem. From this we immediately obtain that every decision problem solvable by a $\PTIME^{\EG}$~algorithm belongs to $\NP \cap \coNP$, and every $\PTIME^{\EG}$~algorithm runs in pseudo-polynomial time, i.e., in time polynomial in $\size{\calE}$, $\size{s(n)}$, and $M_{\calE}$, where $M_{\calE}$ is the maximal absolute value of a counter update in $\calE$. We say that a decision problem \textsf{X} is \emph{$\EG$-hard} if there is a polynomial-time reduction from $\EG$ to \textsf{X}.

Our results (answers to \textbf{[Q1]}--\textbf{[Q3]}) can be formulated as follows: 
\smallskip

\noindent
\textbf{[A1]}. We show that an optimal strategy is guaranteed to exist in a configuration $s(n)$ if the underlying EMDP is \emph{strongly connected and pumpable}. An EMDP is strongly connected if its underlying graph is strongly connected, and pumpable if for every safe configuration $t(m)$ there exists a safe strategy $\sigma$ such that the counter value is unbounded in almost all runs initiated in~$t(m)$. 

The problem whether a given EMDP is strongly connected and pumpable is in $\PTIME^{\EG}$ and $\EG$-hard. Further, an optimal strategy in $s(n)$ does not necessarily exist if just one of these two conditions is violated. We use $\sccp$ to denote the subclass of strongly connected and pumpable EMDPs.
\smallskip

\noindent
\textbf{[A2, A3]}. If a given EMDP belongs to the $\sccp$ subclass, the following holds: 
\begin{itemize}
	\item The value of every safe configuration is the same and computable by a $\PTIME^{\EG}$ algorithm
	  (consequently, the limit value of all states is also the same and computable by a $\PTIME^{\EG}$ algorithm). The gap threshold problem is $\EG$-hard.
	\item There exists a strategy $\sigma$ which is optimal in every configuration. In general, $\sigma$  may require infinite memory. A finite description of $\sigma$ is computable by a $\PTIME^{\EG}$ algorithm. The same holds for $\varepsilon$-optimal strategies where $\varepsilon > 0$, except that $\varepsilon$-optimal strategies require only finite memory.
\end{itemize}
Note that since the gap threshold problem is $\EG$-hard, approximating the value is not much easier than computing the value precisely for $\sccp$s.

For general EMDPs, optimal strategies are not guaranteed to exist. Still, for every EMDP $\calE$ we have the following:
\begin{itemize}
	\item The value of every configuration $s(n)$ can be approximated up to an arbitrarily small given $\varepsilon > 0$ in time polynomial in $\size{\calE}$, $\size{s(n)}$, $\maxE$, and $1/\varepsilon$. The limit value of each control state is computable in time polynomial in  $\size{\calE}$ and $\maxE$.
	\item For a given $\varepsilon > 0$, there exists a strategy $\sigma$ which is \mbox{$\varepsilon$-optimal} in every configuration. In general, $\sigma$  may require infinite memory. A finite description of $\sigma$ is computable in time polynomial in $\size{\calE}$, $\maxE$, and $1/\varepsilon$.
	\item The gap threshold problem is $\PSPACE$-hard.
\end{itemize}

The above results are non-trivial and based on detailed structural analysis of EMDPs. As a byproduct, we yield a good intuitive understanding on what can actually happen when we wish to construct a (sub)optimal strategy in a given EMDP configuration. The main steps are sketched below (we also try to explain where and how we employ the existing ideas, and where we needed to invent original techniques). The details and examples illustrating the discussed phenomena are given later in Section~\ref{sec:onedim}.

The core of the problem is the analysis of maximal end components of a given EMDP, so let us suppose that our EMDP is strongly connected (but not necessarily pumpable). First, we check whether there exists \emph{some} strategy such that the average change of the counter per transition is positive (this can be done by linear programming) and distinguish two possibilities:

\emph{\bfseries If there is such a strategy}, then we try to optimize the mean payoff under the constraint 
	   that the average change of the counter is non-negative. This can be formulated by a linear program whose solution allows to construct finitely many randomized memoryless strategies and an appropriate ``mixing ratio'' for these strategies that produces an optimal mean payoff. This part is inspired by the technique used in \cite{BBCFK14} for the analysis of MDPs with multiple mean-payoff objectives. However, here we cannot implement the optimal mixing ratio ``immediately'' because we also need to ensure that the resulting strategy is safe. We can solve this problem using two different methods, depending on whether the EMDP is pumpable or not. If it is not pumpable, then, since we aim at constructing an \mbox{$\varepsilon$-optimal} strategy, we can always slightly modify the mix, adding the aforementioned strategy which increases the counter in a right proportion. If the counter becomes too low, we permanently switch to some safe strategy (which may produce a low mean payoff). Since the counter has a tendency to increase, we can setup everything so that the probability of visiting low counter values is very small if we start with a sufficiently large initial counter value. Hence, for configurations with a sufficiently large counter value, we play $\varepsilon$-optimally. For the configurations with ``low'' counter value, we compute a suboptimal strategy by ``cutting'' the counter when it reaches a large value (where we already know how to play) and applying the algorithm for finite-state MDPs. 
	   
	   More interesting is the case when the EMDP \emph{is} pumpable. Here, instead of switching to \emph{some} safe strategy, we switch to a \emph{pumping} strategy, i.e. a safe strategy that is capable of increasing the counter above any threshold with probability 1. Once the pumping strategy increases the counter to some sufficiently high value, we can switch back to playing the aforementioned ``mixture.'' To obtain an optimal strategy in this way, we need to extremely carefully set up the events which trigger ``(de-)activation'' of the pumping strategy, so as to ensure that it keeps the counter sufficiently high and at the same time assure that it does not negatively affect the mean payoff. We innovatively use the martingale techniques designed in~\cite{BKK:oc-jacm} to accomplish this delicate task.

\emph{\bfseries If there is no such strategy}, we need to analyze our EMDP differently. We prove that   
       \emph{every} safe strategy then satisfies the following: almost all runs end by an infinite suffix where all visited configurations with the same control state have the same counter value. This implies that only finitely many configurations are visited in the suffix, and we can analyze the associated mean payoff by methods for finite-state MDPs.     

If we additionally assume that our strongly connected EMDP is pumpable, than there inevitably exists a strategy which increases the counter on average (which rules out the second possibility mentioned above) and the ``switching'' strategy can be constructed differently so that it achieves the optimal mean payoff specified by the linear program. 

Let us note that some of the presented ideas can be easily extended even to multi-energy MDPs. Since a full analysis of EMDPs is rather lenghty and complicated, we leave this extension for future work. 
\medskip

\noindent
\textbf{Related Work.}
MDPs with mean payoff objectives (average reward criteria) have been heavily studied since the 60s~(see, e.g., \cite{Howard60,Puterman:book}). Several algorithms for computing optimal values and strategies have been developed for both finite-state systems~(see~e.g.~\cite{Puterman:book,FV96,BBCFK14,CKK15}) as well as various types of infinite-state MDPs typically related to queueing systems (see, e.g.,~\cite{KR95}). For an extensive survey see~\cite{Puterman:book}. 

Markov decision processes with energy objectives have been studied in~\cite{DBLP:conf/soda/BrazdilBEKW10} as one-counter MDPs. Subsequently, several papers concerned MDPs with counters (resources) have been published (for a survey see~\cite{DBLP:conf/rp/Kucera12}, for recent work see e.g.~\cite{ACMSS15}). A closely related paper~\cite{CD:energy-parity-games} studies MDPs with combined energy-parity and mean-payoff-parity objectives (note, however, that the combination of energy with mean payoff is not studied in~\cite{CD:energy-parity-games}). 

A considerable amount of attention has been devoted to non-stochastic turn-based games with energy objectives \cite{DBLP:conf/emsoft/ChakrabartiAHS03,DBLP:conf/formats/BouyerFLMS08}. Solving energy games belongs to $\NPcoNP$ but no polynomial time algorithm is known. Energy games are polynomially equivalent to mean-payoff games~\cite{DBLP:conf/formats/BouyerFLMS08}. Several papers are concerned with complexity of energy games (or equivalent problems, see~e.g.~\cite{GKK90,DBLP:journals/tcs/ZwickP96,DBLP:journals/fmsd/BrimCDGR11,CRR14}). For a more detailed account of results on energy games see~\cite{CHKN14}. Games with various combinations of objectives as well as multi-energy objectives have also been studied (see~e.g.~\cite{DBLP:journals/iandc/VelnerC0HRR15,AMSS13,Brenguier:2014:EMT:2562059.2562116,DBLP:conf/birthday/JuhlLR13,DBLP:conf/fsttcs/ChatterjeeDHR10,CD:energy-parity-games,DBLP:journals/corr/BouyerMRLL15}), as well as energy constraints in automata settings~\cite{CFL:energy-kleene}. %

Our work is closely related to the recent papers~\cite{bruyre_et_al:LIPIcs:2014:4458,Clemente:2015:MBW:2876514.2876539} where the combination of expected and worst-case mean-payoff objectives is considered. In particular,~\cite{Clemente:2015:MBW:2876514.2876539} considers a problem of optimizing the expected multi-dimensional mean-payoff under the condition that the mean-payoff in the first component is positive for all runs. At first glance, one may be tempted to ``reduce''~\textbf{[Q2]}~and~\textbf{[Q3]} to results of~\cite{Clemente:2015:MBW:2876514.2876539} as follows: Ask for a strategy which ensures that the mean-payoff in the first counter is non-negative for all runs, and then try to optimize the expected mean-payoff of the second counter. However, this approach does not work for several reasons. First, a strategy achieving non-negative mean-payoff in the first counter may still decrease the counter arbitrarily deep. So no matter what initial value of the counter is used, the zero counter value may be reached with positive probability. Second, the techniques developed in~\cite{Clemente:2015:MBW:2876514.2876539} do not work in the case of ``balanced'' EMDPs. Intuitively, balanced EMPDs are those where we inevitably need to employ strategies that balance the counter, i.e., the expected average change of the counter per transition is zero. In the framework of stochastic counter systems, the balanced subcase is often more difficult than the other subcases when the counters have a tendency to ``drift'' in some direction. In our case, the balanced EMDPs also require a special (and non-trivial) proof techniques based on martingales and some new ``structural'' observations. We believe that these tools can be adapted to handle the ``balanced subcase'' in even more general problems related to systems with more counters, MDPs over vector addition systems, and similar models.

%% file: prelims.tex
\section{Preliminaries}
\label{sec:prelims}

\noindent
We use $\Zset$, $\Nset$, $\Nset^+$, $\Qset$, and $\Rset$
to denote the set of all integers, non-negative
integers, positive integers, rational numbers, and real numbers, 
respectively. 
We assume familiarity with basic notions of probability theory, e.g.,
\emph{probability space}, \emph{random variable}, or the \emph{expected 
value}. As usual, a \emph{probability distribution} over a finite or
countably infinite set $A$ is a function
$f : A \rightarrow [0,1]$ such that \mbox{$\sum_{a \in A} f(a) = 1$}.
We call $f$ \emph{positive} if
$f(a) > 0$ for each $a \in A$,  \emph{rational} if $f(a) \in
\Qset$ for each $a \in A$, and \emph{Dirac} if $f(a) = 1$ for some $a \in A$.

\begin{definition}[MDP]
A \emph{Markov decision process (MDP)} is a tuple 
$\M = (\states,(\cstates,\stochstates),\trans{},\Prob,\rev)$,
where $\states$ is a finite set of \emph{states}, $(\cstates,\stochstates)$
is a 
partitioning of $\states$ into the sets $\cstates$ of  \emph{controllable}
states and $\stochstates$ of 
 \emph{stochastic} states, respectively, $\trans{} \subseteq \states \times
\states$ is a \emph{transition relation}, 
$\Prob$ is a function assigning to every stochastic state $s \in \stochstates$
a positive probability distribution over its outgoing transitions, and
$\rev\colon\trans \rightarrow \Qset$ is a \emph{reward function}. We assume that
$\trans{}$ is \emph{total}, i.e., for each $s\in \states$ there is $t \in
\states$ such that $(s,t)\in \trans$.
\end{definition}

We use $\Prob(s,t)$ as an abbreviation for $(\Prob(s))(s,t)$, i.e., $\Prob(s,t)$
is the probability of taking the transition $(s,t)$ in $s$. 
For a state $s$ we denote by $\out(s)$ the set of transitions outgoing from $s$.
A \emph{finite path} is a sequence $w=s_0 s_1 \cdots s_n$ of states
such that $(s_i, s_{i+1})\in \trans$ for all $0\leq i < n$. We write
$\len{w}=n$ for the length of the path.  A \emph{run} (or an \emph{infinite path}) is an
infinite sequence $\omega$ of states such that every finite prefix 
of~$\omega$ is a finite path.
For a finite path $w$, we denote by $\runs{w}$ the set of all runs having
$w$ as a prefix.  

An \emph{end component} of $\calM$ is a pair $(S',T')$, where $S' \subseteq S$, $T' \subseteq T$, satisfying the following conditions: (1) for every $s \in S'$, we have that $\out(s) \cap T' \neq \emptyset$; (2) if $s \in S' \cap \stochstates$, then $\out(s) \subseteq T'$; (3) the graph determined by $(S',T')$ is strongly connected. Note that every end component of $\calM$ can be seen as a strongly connected MDP (obtained by restricting the states and transitions of $\calM$). A \emph{maximal end component (MEC)} is an end component which is maximal w.r.t.{} pairwise inclusion. The MECs of a given MDP  $\calM$ are computable in polynomial time~\cite{CH14}.

A \emph{strategy} (or a \emph{policy}) in an MDP $\M$ is a tuple $\sigma = (M,m_0,\update,\nxt)$ where $M$ is a set of memory elements, $m_0 \in M$ is an initial memory element, $\update : M \times S \rightarrow M$ a memory-update function, and $\nxt$ is a function which to every pair $(s,m) \in \cstates \times M$ assigns a probability distribution over $\out(s)$. The function $\update$ is extended to finite sequences of states in the natural way. We say that $\sigma$ is \emph{finite-memory} if $M$ is finite, and \emph{memoryless} if $M$ is a singleton. Further, we say that $\sigma$ is \emph{deterministic} if $\nxt(s,m)$ is Dirac for all $(s,m) \in \cstates \times M$. Note that $\sigma$ determines a function which to every finite path in $\M$ of the form $w s$, where $s \in \cstates$, assigns the probability distribution $\nxt(s,m)$, where $m = \update(m_0,w)$. Slightly abusing our notion, we use $\sigma$ to denote this function.

Fixing a strategy $\sigma$ and an initial state $s$, we obtain the standard probability space $(\run_{\M}(s),\calF,\Probm^{\sigma}_{s})$ of all runs starting at $s$, where
$\calF$ is the \mbox{$\sigma$-field} generated by all \emph{basic cylinders}
$\run_{\M}(w)$, where $w$ is a finite path starting at~$s$, and
$\Probm^{\sigma}_{s}\colon \calF \rightarrow [0,1]$ is the unique probability measure such that for all finite paths $w=s_0\cdots s_n$ it holds 
$\Probm^{\sigma}_{s}(\run_{\M}(w)) = \prod_{i{=}1}^{n} x_i$, where
each $x_i$ is either $\sigma(s_0\cdots s_{i-1})(s_{i-1},s_{i})$, or $\Prob(s_{i-1},s_{i})$, depending on whether $s_{i-1}$ is controllable or stochastic (the empty product evaluates to 1). We denote by $\E^{\sigma}_{s}$ the expectation operator of this probability space.

We say that a run $\omega = s_0s_1\cdots$ is \emph{compatible} with a strategy $\sigma$ if $\sigma(s_0\cdots s_i)(s_i,s_{i+1})>0$ for all $i\geq 0$ such that $s_i \in \cstates$.

\begin{definition}[EMDP]
An \emph{energy MDP (EMDP)} is a tuple
$\calE=(\M,\energy)$, where $\M$ is a finite MDP and $\energy$ is a function
assigning to every transition an integer \emph{update}.
\end{definition}

We implicitly extend all MDP-related notions to EMPDs, i.e., for $\calE=(\M,\energy)$ we speak about runs and strategies in $\calE$ rather than about runs and strategies in $\M$. A \emph{configuration} of $\calE$ is an element of $\states \times \Zset$ written as $s(n)$.

Given an EMDP $\calE=(\M,\energy)$ and a configuration $s(n)$ of $\calE$, we use $\size{\calE}$ and $\size{s(n)}$ to denote the encoding size of $\calE$ and  $s(n)$, respectively, where the counter updates and rewards used in $\calE$, as well as the $n$ in $s(n)$, are written as (fractions of) binary numbers. We also use $\maxE$ to denote the maximal non-negative integer $u$ such that $u$ or $-u$ is an update assigned by $E$ to some transition. 

Given a finite or infinite path $w=s_0s_1\cdots$ in $\calE$ and an \emph{initial configuration} $s_0(n_0)$, we define the \emph{energy level} after $i$ steps of $w$ as $\enlev{n_0}{w}{i} = n_0 + \sum_{i=0}^{i-1}\energy(s_i,s_{i+1})$ (the empty sum evaluates to zero). A configuration of $\calE$ after $i$ steps of $w$ is then the configuration $s_i(n_i)$, where $n_i = \enlev{n_0}{w}{i}$. Note that for all $n$ and $i\geq 0$, $\enlev{n}{}{i}$ can be understood as a random variable.

We say that a run $\omega$ initiated in $s_0$ is \emph{safe} in a configuration $s_0(n_0)$ if $\enlev{n_0}{w}{i}\geq 0$ for all $i\geq 0$. A strategy $\sigma$ is safe in $s_0(n_0)$ if all runs compatible with $\sigma$ are safe in $s_0(n_0)$. Finally, a configuration $s_0(n_0)$ is safe if there is at least one strategy safe in $s_0(n_0)$ . The following lemma is straightforward.

\begin{lemma}
\label{lem:energy-monotone}
If $s(n)$ is safe and $m \geq n$, then $s(m)$ is safe.
\end{lemma}

To every run $\omega=s_0s_1\cdots$ in $\calE$ we assign a mean payoff $\MP(\omega)$ collected along $\omega$ defined as
$\MP(\omega) := \liminf_{n\rightarrow \infty} (\sum_{i=1}^{n}\rev(s_{i-1},s_{i}))/n$. The function $\MP$ can be seen as a random variable, and for every strategy $\sigma$ and initial state $s$ we denote by $\E^{\sigma}_{s}[\MP]$ its expected value (w.r.t.{} $\Probm^{\sigma}_{s}$).

\begin{definition}[Energy-constrained value]
Let $\calE=(\M,\energy)$ be an EMDP and $s(n)$ its configuration. The \emph{energy-constrained mean-payoff value} (or simply the \emph{value}) of $s(n)$ is defined by
\(
  \val(s(n)) := \sup\ \{ \E^{\sigma}_{s}[\MP]  \mid \sigma \mbox{ is safe in } s(n)\} \,.
\)
For every state $s$ we also put $\val(s) := \lim_{n \rightarrow \infty} \val(s(n))$.
\end{definition}

Note that the value of every unsafe configuration is $-\infty$.
We say that a strategy $\sigma$ is \emph{$\eps$-optimal} in  $s(n)$, where $\eps\geq 0$, if $\sigma$ is safe in $s(n)$ and $\val(s(n))-\E^{\sigma}_{s}[\MP]\leq \eps$. A $0$-optimal strategy is called \emph{optimal}.

%% file: one-counter.tex
\section{The Results}
\label{sec:onedim}

In this section we precisely formulate and prove the results about EMDPs announced in Section~\ref{sec-intro}. Let $\calE=(\calM,\energy)$ be an EMDP. For every state $s$ of $\calE$, let $\minsafe(s)$ be the least $n \in \Nset$ such that $s(n)$ is a safe configuration.  If there is no such $n$, we put $\minsafe(s)=\infty$. The following lemma follows from the standard results on one-dimensional energy games~\cite{CdAHS:resource-interfaces}.

\begin{lemma}
\label{lem:minsafe-onedim}
There is a $\PTIME^{\EG}$~algorithm which computes, for a given EMDP $\calE=(\calM,\energy)$ and its state $s$, the value $\minsafe(s)$. 
\end{lemma}

Next, we present a precise definition of strongly connected and pumpable EMPDs. We say that $\calE$ is \emph{strongly connected} if for each pair of states $s,t$ there is a finite path starting in $s$ and ending in $t$. The pumpability condition is more specific.

\begin{definition}
	\label{def-pumpable}
	Let $\calE$ be an EMDP and $s(n)$  a configuration of~$\calE$. We say that a strategy $\sigma$ is \emph{pumping in $s(n)$} if $\sigma$ is safe in $s(n)$ and $\Probm_s^{\sigma}(\sup_{i\geq 0} \enlev{n}{}{i}=\infty)=1$. Further, we say that $s(n)$ is \emph{pumpable} if there is a strategy pumping in $s(n)$, and $\calE$ is \emph{pumpable} if every safe configuration of $\calE$ is pumpable.
\end{definition}

The subclass of strongly connected pumpable EMDPs is denoted by $\sccp$. 
Clearly, if $s(n)$ is pumpable, then every $s(m)$, where $m\geq n$, is also pumpable. Hence, for every  $s\in \states$, we define $\minpump(s)$ as the least $n$ such that $s(n)$ is pumpable. If there is no such $n$, we put \mbox{$\minpump(s)=\infty$}. 

Intuitively, the condition of pumpability allows to increase the counter to an arbitrarily high value whenever we need. The next lemma says that we can compute a strategy which achieves that. 

\begin{lemma}
	\label{lem:minpump}
	For every EMDP $\calE$ there exist a memoryless \emph{globally pumping} strategy $\sigma$, i.e. a strategy that  is pumping in every pumpable configuration of $\calE$. Further, there is a $\PTIME^{\EG}$~algorithm which computes the strategy $\sigma$ and the value $\minpump(s) \leq 3\cdot |\states|\cdot \maxE$ for every state $s$ of $\calE$. The problem whether a given configuration of $\calE$ is pumpable is $\EG$-hard.
\end{lemma}

\noindent
Now we can state our results about $\sccp$s. 

\begin{theorem}
\label{thm:onedim-scc-pump}
For the subclass of $\sccp$s, we have the following:
\begin{enumerate}
\item The problem whether a given EMDP $\calE$ belongs to $\sccp$ is $\EG$-hard and solvable by a $\PTIME^{\EG}$~algorithm. 
\item The value of all safe configurations of a given $\sccp$ $\calE$ is the same. Moreover, there is a $\PTIME^{\EG}$~algorithm which computes this value.
\item For every $\sccp$  $\calE$  and every configuration $s(n)$ of~$\calE$, there is a strategy $\sigma$ optimal in $s(n)$. In general, $\sigma$ may require infinite memory, and there is a $\PTIME^{\EG}$~algorithm which computes a finite description of this strategy.
\item For every $\sccp$  $\calE$, every configuration $s(n)$ of~$\calE$, and every $\varepsilon > 0$, there is a finite-memory strategy which is $\varepsilon$-optimal in $s(n)$. Further, there is a $\PTIME^{\EG}$~algorithm which computes a finite description of this strategy.
\item The gap threshold problem for $\sccp$s is $\EG$ hard.
\end{enumerate}
\end{theorem}

\noindent
In particular, note that $\varepsilon$-optimal strategies in $\sccp$s require only finite memory~(4.), but they are not easier to compute than optimal strategies~(5.).

The following theorem summarizes the results for general EMDPs.

\begin{theorem}
\label{thm:onedim-general}
For general EMDPs, we have the following:
\begin{enumerate}
\item Optimal strategies may not exist in EMDPs that are either not strongly connected or not pumpable.
\item Given an EMDP $\calE$, a configuration $s(n)$ of $\calE$, and $\varepsilon > 0$, the value of $s(n)$ can be approximated up to the absolute error $\varepsilon$ in time which is polynomial in $\size{\calE}$, $\size{s(n)}$, $M_{\calE}$, and $1/\varepsilon$. 
\item Given an EMDP $\calE$ and a state $s$ of $\calE$, the limit value $\val(s)$ is computable in time polynomial in $\size{\calE}$ and $M_{\calE}$.
\item Let $\calE$ be an EMDP, $s(n)$ a configuration of $\calE$, and $\varepsilon > 0$. An $\varepsilon$-optimal strategy in $s(n)$ may require infinite memory. A finite description of a strategy $\sigma$ which is $\varepsilon$-optimal strategy in $s(n)$ is computable in time polynomial in $\size{\calE}$, $M_{\calE}$, and $1/\varepsilon$. 
\item The gap threshold problem for EMDPs is in $\EXPTIME$ and $\PSPACE$-hard.
\end{enumerate}
\end{theorem}

%% file: results-tech-petr.tex
Before proving Theorems~\ref{thm:onedim-scc-pump} and \ref{thm:onedim-general}, we introduce several tools that are useful for the analysis of strongly connected EMDPs.
For the rest of this section, we fix a \emph{strongly connected} EMDP $\calE=(\M,\energy)$ where $\M = (\states,(\cstates,\stochstates),\trans{},\Prob,\rev)$.

The key component for the analysis of $\calE$ is the linear program $\calL_{\calE}$ shown in Figure~\ref{fig:oc-lp}~(left). The program is a modification of a program used in~\cite{BBCFK14}  for multi-objective mean-payoff optimization. For each transition $e$ of $\calE$ we have a non-negative variable $f_e$ that intuitively represents the long-run frequency of traversals of $e$ under some strategy (the fact that $f_e$'s can be given this interpretation is ensured by the \emph{flow constraints} introduced in the first three lines). The constraint on the fourth line then ensures that a strategy that visits each transition $e$ with frequency $f_e$ achieves a non-negative long-run change of the energy level. In other words, such a strategy ensures that the energy level does not have, on average, a tendency to decrease.

Intuitively, the optimal value of $\calL_{\calE}$ is the maximal expected mean payoff achievable under the constraint that the long-run average change (or \emph{trend}) of the energy level is non-negative. Every safe strategy has to satisfy this constraint, because otherwise the probability of visiting a configuration with negative counter would be positive. Thus, using the methods adopted from~\cite{BBCFK14}, we get the following.

\begin{lemma}
\label{lem:strategy-to-lpsol}
If there is a strategy $\sigma$ that is safe in some configuration $s(n)$ of~$\calE$, then the linear program $\calL_{\calE}$ has a solution whose objective value is at least $\E^{\sigma}_{s}[\MP]$. 
\end{lemma}

\begin{figure}
\parbox{.45\textwidth}{%
\centering\small
\begin{align*}
\intertext{\textbf{maximize} $\quad \sum_{e \in \trans} f_e \cdot \rev(e)\quad$  subject to}
\\[-3ex]
\sum_{e \in \trans} f_e &= 1
\\
\text{$\forall s \in \cstates$:} \sum_{(s',s)\in\trans} f_{(s',s)} &= \sum_{(s,s'')\in \trans} f_{(s,s'')}
\\
\text{$\forall s \in \stochstates$, $\forall (s,r)\in \trans$:}\; f_{(s,r)} &= \Prob(s,r)\cdot\sum_{(s',s)\in \trans} f_{(s',s)}
\\
\sum_{e \in \trans} f_e \cdot \energy(e) &\geq 0
\\
\text{$\forall e\in \trans$:}\; f_e &\geq 0
\end{align*}}
\hfill
\parbox{.45\textwidth}{%
\centering\small
\begin{tikzpicture}[x=1cm,y=1.5cm,font=\footnotesize]
\node (A1) at (0,0) [ran] {$s$};
\node (A2) at (0,-1.5) [ran] {$t$};
\draw (A1) [tran, rounded corners] -- +(-.3,0.5) -- node[above] {$\cu{2};0;\frac{1}{2}$} +(.3,.5) -- (A1);
\draw (A2) [tran, rounded corners] -- +(-.3,-.5) -- node[below] {$\cu{-1};0;\frac{1}{2}$} +(.3,-.5) -- (A2);
\draw (A1) [tran, bend left=30] to node[right] {$\cu{0};0;\frac{1}{2}$} (A2);
\draw (A2) [tran, bend left=30] to node[left] {$\cu{0};0;\frac{1}{2}$} (A1);
\end{tikzpicture}}
\caption{A linear program $\calL_{\calE}$ with non-negative variables $f_e$, $e\in \trans$ (left), and an EMDP where the strategy corresponding to the solution of $\calL_{\calE}$ is not safe (right).}
\label{fig:oc-lp}
\end{figure}
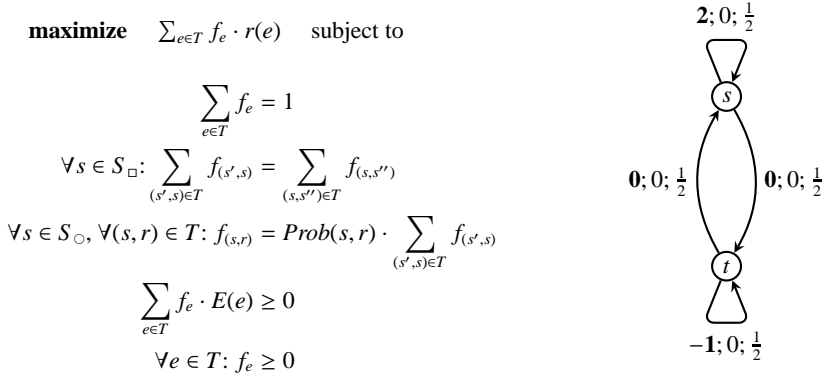

On the other hand, even if a strategy achieves a non-negative (or even positive) counter trend, it can still be unsafe in all configurations of $\calE$. To see this, consider the EMDP of Figure~\ref{fig:oc-lp}~(right). There is only one strategy (the empty function), and it is easy to verify that assigning $1/4$ to each variable in $\calL_{\calE}$ solves the linear program with objective value $1/2$. However, for every $m$ there is a positive probability that the decrementing loop on $s$ is taken at least $m$ times, and thus the strategy is not safe.

Although the program $\calL_{\calE}$ cannot be directly used to obtain a safe strategy optimizing the mean payoff, it is still useful for obtaining certain ``building blocks'' of such a strategy. To this end, we introduce additional terminology. 

Let $\vec{f}=(f_e)_{e\in\trans}$ be an optimal solution of $\calL_{\calE}$, and let $f^*$  be the corresponding optimal value of the objective function. A \emph{flow graph} of $\vec{f}$ is a digraph $G_{\vec{f}}$ whose vertices are the states of $\calE$, and there is an edge $(s,t)$ in $G_{\vec{f}}$  iff there is a transition $e=(s,t)$ with $f_e>0$. 
A \emph{component} of $\vec{f}$ is a maximal set $C$ of states that forms a strongly connected subgraph of $G_{\vec{f}}$. The set $\trans_C$ consists of all $(s,t)\in \trans$ such that $s\in C$ and $f_{(s,t)}>0$.
A \emph{frequency} of a component $C$ is the number $\freq{C} = \sum_{e\in \trans_C} f_{e}$. Finally, a \emph{trend} and \emph{mean-payoff} of a component $C$ are the numbers $\trend_C = \sum_{e\in \trans_C} ({f_e}/{\freq{C}})\cdot\energy(e)$ and $\meanp_C = \sum_{e\in \trans_C} ({f_e}/{\freq{C}})\cdot\rev(e)$. 

Intuitively, the components of $\vec{f}$ are those families of states that are visited infinitely often by a certain strategy that maximizes the mean payoff while ensuring that the counter trend is non-negative. We show that our analysis can be simplified by considering only certain components of $\vec{f}$. We define a \emph{type I core} and \emph{type II core} of $\vec{f}$ as follows:

\begin{itemize}
\item A type I core of $\vec{f}$ is a component $C$ of $\vec{f}$ such that $\trend_C > 0$ and $\meanp_C\geq f^*$.
\item A type II core of $\vec{f}$ is a pair $C_1$, $C_2$ of its components such that $\trend_{C_1}\geq 0$, $\trend_{C_2}\leq 0$, $\freq{C_1}\cdot \trend_{C_1} + \freq{C_2}\cdot \trend_{C_2}\geq 0$ and $\freq{C_1}\cdot \meanp_{C_1} + \freq{C_2}\cdot \meanp_{C_2}\geq f^* $.
\end{itemize}

\noindent
The following lemma is easy.

\begin{lemma}
\label{lem:solution-core}
Each optimal solution $\vec{f}$ of $\calL_{\calE}$ has a type I or a type II core. Moreover, a core of $\vec{f}$ (of some type) can be found in polynomial time.
\end{lemma}

\subsection{Strongly Connected and Pumpable EMDPs}
\label{sec-SPEMDPs}

In this subsection, we continue our analysis under the assumption that the considered EMPD~$\calE$ is not only strongly connected but also pumpable. Let $\vec{f}$ be an optimal solution to $\calL_{\calE}$ with optimal value $f^*$. We show how to use $\vec{f}$ and its core to construct a strategy optimal in every configuration $s(n)$ of $\calE$. To some degree, the construction depends on the type of the core we use.

We start with the easier case when we compute a type~I core $C$ of $\vec{f}$. Consider two memoryless strategies: First, a memoryless deterministic globally pumping strategy $\pi$ which is guaranteed to exist by Lemma~\ref{lem:minpump}. Second, we define a memoryless randomized strategy $\mu_C$ such that $\mu_C(s)(e)= f_{e}/\freq{C}$ for all $s\in C$ and $e \in \out(s)$, and $\mu_C(s)(e)=\kappa(s)(e)$ for all  $s\not\in C$ and $e \in \out(s)$, where $\kappa$ is a memoryless deterministic strategy in $\calE$ ensuring that a state of $T$ is reached with probability~$1$ (such a strategy exists as $\calE$ is strongly connected).
In order to combine these two strategies, we define a function $\test{n}$ which assigns to a finite path $w$ a value 1 if and only if there is $0\leq j \leq \len{w}$ such that $\enlev{n}{w}{j}\leq L:=\maxE+\max_{s\in \states} \minpump(s)$ and $\enlev{n}{w}{i}\leq H:=L+|\states|+2|\states|^2\cdot\maxE$ for all $j\leq i \leq \len{w}$; otherwise, $\test{n}(w)=0$.
We then define a strategy $\sigma^*_n$ as follows:

\[
\sigma^*_n(w)(e) = \begin{cases}
\mu_C(\last(w))(e) & \text{if $\test{n}(w)=0$}\\
\pi(\last(w))(e) & \text{if $\test{n}(w)=1$}.
\end{cases}
\]

\begin{proposition}
\label{prop:typeA-optimality}
Let $s(n)$ be a configuration of $\calE$. Then $\sigma^*_n$ is optimal in $s(n)$.
\end{proposition}

Let us summarize the intuition behind the proof of Proposition~\ref{prop:typeA-optimality}. If the counter value is sufficiently high, we play the strategy $\mu$ prescribed by $\calL_{\calE}$ (i.e., we strive to achieve the mean payoff value $f^*$) until the counter becomes ``dangerously low'', in which case we switch to a pumping strategy that increases the counter to a sufficiently high value, where we again switch to $\mu$. The positive counter trend achieved by $\mu$ ensures that if we start with a sufficiently high counter value, the probability of the counter \emph{never} decreasing to dangerous levels is bounded away from zero. Moreover, once we switch to the pumping strategy $\pi$, with probability 1 we again pump the counter above $|\states|\cdot H$ and thus switch back to $\mu$. Hence, with probability 1 we eventually switch to strategy $\mu$ and use this strategy forever, and thus achieve mean payoff $f^*$. 

Let us now consider the case where we compute a type~II core of $\vec{f}$.  The overall idea is similar as in the type I case. We try to execute a strategy that has non-negative counter trend and achieves the value $f^*$ computed by $\calL_{\calE}$. This amounts to periodical switching between components $C_1$ and $C_2$, in such a way that the ratio of time spent in $C_i$ tends to $\freq{C_i}$. As in~\cite{BBCFK14}, this is done by fixing a large number $N$ and fragmenting the play into infinitely many iterations: in the \mbox{$k$-th} iteration, we spend roughly $k\cdot N\cdot \freq{C_1}$ steps in $C_1$, then move to $C_2$ and spent $k\cdot N\cdot \freq{C_2}$ steps in $C_2$, then move back to $C_1$ and initialize the \mbox{$(k{+}1)$-th} iteration. Inside the component $C_i$ we use the strategy $\mu_{C_i}$ defined above, until it either is time to switch to $C_{3-i}$ or the counter becomes dangerously low. If the latter event happens, we immediately end the current iteration, switch to a pumping strategy, wait until a counter increases to a sufficient height, and then begin the \mbox{$(k{+}1)$-th} iteration. However, as the trend of $\mu_{C_2}$ is negative, the energy level tends to return to the value to which we increase the level during the pumping phase: it is thus no longer possible to prove, that we eventually stop hitting dangerously low levels. To overcome this problem, we use \emph{progressive pumping}: the height to which we want to increase the counter after the ``pumping mode'' is switched on in the \mbox{$k$-th} iteration must increase with $k$, and it must increase asymptotically faster than $\sqrt{k}$. If this technical requirement is satisfied, we can use martingale techniques to show that progressive pumping decreases, with each iteration, the probability of drops towards dangerous levels. However, it also lengthens the time spent on pumping once such a period is initiated. To ensure that the fraction of time spent on pumping still tends to zero, we have to ensure that the threshold to which we pump increases \emph{sublinearly} in $k$. In our proof we set the bound to roughly $k^{\frac{3}{4}}$ in order to satisfy both of the aforementioned constraints. More details in the appendix.

\begin{proposition}
\label{prop:type-II-optimality}
Each type II core of $\vec{f}$ yields a strategy optimal in $s(n)$. 
\end{proposition}

%% file: thm2-proof.tex
\subsection{General EMDPs}
\label{sec-thm2}

In this section we prove Theorem~\ref{thm:onedim-general}. The two counterexamples required to prove part~(1.) of the theorem are given in Fig.~\ref{fig-not-optimal}. On the left, there is a strongly connected but not pumpable EMDP (note that $t(0)$ is safe but not pumpable) where $\val(s(0)) = 5$, but there is no optimal strategy, and \emph{every} strategy achieving a positive mean-payoff requires infinite memory (hence, this example also  demonstrates that $\varepsilon$-optimal strategies may require infinite memory, as stated in part~(4) of Theorem~\ref{thm:onedim-general}). This is because the counter must be pumped to \emph{linearly} larger and larger values when revisiting $s$ to avoid reaching the configuration $t(0)$ with probability one (note that the probability of visiting $t(0)$ from $t(N)$ when using the transition $(t,u)$ decays \emph{exponentially} in $N$), yet ensuring that the mean payoff is equal to~$5$. Also note that if the counter was pumped to \emph{exponentially} larger and larger values when revisiting $s$, the defining $\liminf$ of mean payoff would be zero.   On the right, there is pumpable but not strongly connected EMDP where $\val(a(0)) = 5$, but no optimal strategy exists in $a(0)$. 

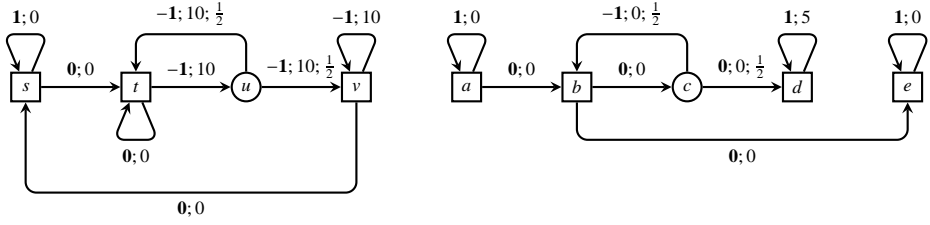
\begin{figure*}[t]
	\begin{center}
		\begin{tikzpicture}[x=1.45cm,y=1.4cm,font=\scriptsize]
		\node (A1) at (0,0) [max] {$s$};
		\node (A2) at (1,0) [max] {$t$};
		\node (A3) at (2,0) [ran] {$u$};
		\node (A4) at (3,0) [max] {$v$};
		\draw (A1) [tran, rounded corners] -- +(.2,.5) -- node[above] {$\cu{1};0$} +(-.2,.5) -- (A1);
		\draw (A1) [tran, rounded corners] -- node[above] {$\cu{0};0$} (A2);
		\draw (A2) [tran, rounded corners] -- node[above] {$\cu{-1};10$} (A3);
		\draw (A2) [tran, rounded corners] -- +(.2,-.5) -- node[below] {$\cu{0};0$} +(-.2,-.5) -- (A2);
		\draw (A3) [tran, rounded corners] -- node[above] {$\cu{-1};10; \frac{1}{2}$} (A4);
		\draw (A3) [tran, rounded corners] -- +(0,.5) -- node[above] {$\cu{-1};10; \frac{1}{2}$}  +(-1,.5) -- (A2);
		\draw (A4) [tran, rounded corners] -- +(0,-1) -- node[below] {$\cu{0};0$}  +(-3,-1) -- (A1);
		\draw (A4) [tran, rounded corners] -- +(.2,.5) -- node[above] {$\cu{-1};10$} +(-.2,.5) -- (A4);
		\node (B1) at (4,0) [max] {$a$};
		\node (B2) at (5,0) [max] {$b$};
		\node (B3) at (6,0) [ran] {$c$};
		\node (B4) at (7,0) [max] {$d$};
		\node (B5) at (8,0) [max] {$e$};
		\draw (B1) [tran, rounded corners] -- +(.2,.5) -- node[above] {$\cu{1};0$} +(-.2,.5) -- (B1);
		\draw (B1) [tran, rounded corners] -- node[above] {$\cu{0};0$} (B2);
		\draw (B2) [tran, rounded corners] -- node[above] {$\cu{0};0$} (B3);
		\draw (B3) [tran, rounded corners] -- node[above] {$\cu{0};0; \frac{1}{2}$} (B4);
		\draw (B3) [tran, rounded corners] -- +(0,.5) -- node[above] {$\cu{-1};0; \frac{1}{2}$}  +(-1,.5) -- (B2);
		\draw (B2) [tran, rounded corners] -- +(0,-.5) -- node[below] {$\cu{0};0$}  +(3,-.5) -- (B5);
		\draw (B4) [tran, rounded corners] -- +(.2,.5) -- node[above] {$\cu{1};5$} +(-.2,.5) -- (B4);
		\draw (B5) [tran, rounded corners] -- +(.2,.5) -- node[above] {$\cu{1};0$} +(-.2,.5) -- (B5);
		\end{tikzpicture}
	\end{center}
	\caption{Examples of EMDPs where optimal strategies do not exist in some configurations. Each transition is labeled by the associated counter update (in boldface), reward, and probability (only for the stochastic states $u$ and $c$).}
	\label{fig-not-optimal}
\end{figure*}

For the rest of this section, we fix an EMDP $\calE=(\calM,\energy)$. For simplicity, we assume that \emph{for every $s \in S$ there is some $n \in \Nset$ such that the configuration $s(n)$ is safe}. The other control states can be easily recognized and eliminated (see Lemma~\ref{lem:minsafe-onedim}). 

Since $\calE$ is not necessarily strongly connected, we start by identifying and constructing the MECs of $\calE$ (this can be achieved in time polynomial in $\size{\calE}$). Recall that each MEC of $\calE$ can be seen as an EMDP, and each run eventually stays in some MEC~\cite{Alfaro:thesis}. Hence, we start by analyzing the individual MECs separately. Technically, we first assume that $\calE$ is strongly connected.

\paragraph{The case when $\calE$ is strongly connected.}

Consider a linear program $\calT_{\calE}$ which is the same as the program $\calL_{\calE}$ of Fig.~\ref{fig:oc-lp} except for its objective function which is set to \mbox{\textbf{maximize} $\sum_{t \in \trans} f_t \cdot \energy(t)$}. In other words, $\calT_{\calE}$ tries to maximize the long-run average change of the energy level under the constraints given in $\calL_{\calE}$. Let $\vec{g}=(g_e)_{e\in\trans}$ be an optimal solution of $\calT_{\calE}$, and let $g^*$  be the corresponding optimal value of the objective function. Now we distinguish two cases, which require completely different proof techniques.
\smallskip

\noindent
\textbf{Case~A.} $g^* > 0$.\\
\textbf{Case~B.} $g^* = 0$.
\smallskip

We start with \textbf{Case~A}. 
Note that if $g^* > 0$, then there exists a component $D$ of $\vec{g}$ such that $\trend_D \geq g^* > 0$. We proceed by solving the linear program $\calL_{\calE}$ of Fig.~\ref{fig:oc-lp}, and identifying the core of an optimal solution $\vec{f}$ of $\calL_{\calE}$. 
Recall that $\vec{f}$ can have either a type~I core $C$, or a type~II core $C_1,C_2$. In the first case, we set $E_1 := C$ and $E_2 := C$, and in the latter case we set $E_1 := C_1$ and $E_2 := C_2$. Let us fix some $\varepsilon > 0$. We compute positive rationals $\alpha_1,\alpha_2$ such 
\begin{itemize}
	\item $\alpha_1 + \alpha_2 = 1$
	\item $\alpha_1 \cdot \meanp_{E_1} + \alpha_2 \cdot \meanp_{E_2} \geq f^* -\varepsilon/2$
	\item $\alpha_1 \cdot \trend_{E_1} + \alpha_2 \cdot \trend_{E_2} > 0$.
\end{itemize}
Observe that we can compute $\alpha_1,\alpha_2$ so that the length of the binary encoding of all of the above numbers is polynomial in $\size{\calE}$ and $\size{\varepsilon}$. Now we construct a strategy which 
is safe and $\varepsilon$-optimal in every configuration with a sufficiently high counter value. Intuitively, we again just combine the two memoryless randomized strategies extracted from  $\vec{f}$ (and possibly $\vec{g}$) in the ratio given by $\alpha_1$ and $\alpha_2$. Since the counter now has a tendency to increase under such a strategy, the probability of visiting a ``dangerously low'' counter value can be made arbitrarily small by starting sufficiently high (exponential height is sufficient for the probability to be smaller than $\eps$). Hence, when such a dangerous situation occurs, we can permanently switch to \emph{any} safe strategy (this is where our approach bears resemblance to~\cite{Clemente:2015:MBW:2876514.2876539}). For the finitely many configurations where the counter height is not ``sufficiently large,'' the $\varepsilon$-optimal strategy can be computed by encoding these configurations into a finite MDP and optimizing mean-payoff in this MDP using standard methods. %
 
\begin{figure}[t]
	\begin{center}
		\begin{tikzpicture}[x=3cm,y=1.2cm,font=\footnotesize]
		\node (A1) at (0,0) [max] {$s$};
		\node (A2) at (1,0) [ran] {$t$};
		\draw (A1) [tran, rounded corners] -- +(-.3,.2) -- node[left] {$\cu{0};0$} +(-.3,-.2) -- (A1);
		\draw (A1) [tran, rounded corners] -- node[above] {$\cu{0};0$} (A2);
		\draw (A2) [tran, rounded corners] -- +(0,.7) -- node[above] {$\cu{-1};10;\frac{1}{2}$}  +(-1,.7) -- (A1);
		\draw (A2) [tran, rounded corners] -- +(0,-.7) -- node[below] {$\cu{1};10;\frac{1}{2}$}  +(-1,-.7) -- (A1);
		\end{tikzpicture}
	\end{center}
	\caption{An EMDP where the solution of $\calL_{\calE}$ is irrelevant.}
	\label{fig-no-program}
\end{figure}
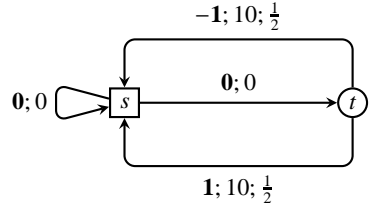

Now consider \textbf{Case~B}. If $g^* = 0$, the solution of $\calL_{\calE}$ is irrelevant, and we need to proceed in a completely different way. To illustrate this, consider the simple EMDP of Fig.~\ref{fig-no-program}. Here, the optimal solution $\vec{f}$ of $\calL_{\calE}$ produces $f^* = 5$ and assigns $1$ to the transition  $(s,t)$. Clearly, we have that $\val(s(n)) = 0$ for an arbitrarily large~$n$, so we cannot aim at approaching $f^*$. Instead, we show that if $g^* = 0$, then almost all runs produced by a safe strategy are \emph{stable} in the following sense. We say that $s \in S$ is \emph{stable at~$k \in \Zset$} in a run $\omega = s_0 s_1 \cdots$ if there exists $i \in \Nset$ such that for every $j \geq i$ we have that $s_j = s$ implies $\enlev{0}{}{j} = k$. Further, we say that $s$ is \emph{stable} in $\omega$ if $s$ is stable at $k$ in $\omega$ for some~$k$. Note that the initial value of the counter does not influence the (in)stability of $s$ in $\omega$. Intuitively, $s$ is stable in $\omega$ if it is visited finitely often, or it is visited infinitely often but from some point on, the energy level is the same in each visit. We say that a \emph{run} is stable if each control state is stable in the run. 

The next proposition represents another key insight into the structure of EMDPs. The proof is non-trivial and can be found in Appendix~\ref{app-thm2}.
\begin{proposition}
\label{lem-stable}
	Suppose that $g^* = 0$, and let $\sigma$ be a strategy which is safe in $s(n)$. Then 
	\[
	  \Probm^{\sigma}_{s}(\{\omega \in \run(s) \mid \omega \mbox{ is stable }\}) = 1 \,.
	\] 
\end{proposition}

Due to Proposition~\ref{lem-stable}, we can analyze the configurations of $\calE$ in the following way. We construct a finite-state MDP where the states are the configurations of $\calE$ with a non-negative counter value bounded by $|S| \cdot \maxE$. Transition attempting to decrease the counter below zero or increase the counter above $|S| \cdot \maxE$ lead to a special sink state with a self-loop whose reward is strictly smaller than the minimal reward used in $\calE$. Then, we apply the standard polynomial-time algorithm for finite-state MDPs to compute the values in the constructed MDP, and identify a configuration $r(\ell)$ with the largest value. By applying  Proposition~\ref{lem-stable}, we obtain that $\val(t) = \val(r(\ell))$ for \emph{every} $t \in S$. For every $\varepsilon > 0$, we can easily compute a bound $N_\varepsilon \in \Nset$  polynomial in $\size{\calE}$, $M_{\calE}$, and $1/\varepsilon$, and a memoryless strategy $\varrho$ such that for every configuration $t(m)$ where $m \geq N_\varepsilon$ we have that the $\Probm^{\varrho}_{t}$ probability of all runs initiated in $t(m)$ that visit a configuration $r(k)$ for some $k \geq \ell$ without a prior visit to a configuration where the counter is ``dangerously low'' is at least $1 - (\varepsilon/R)$, where $R$ is the difference between the maximal and the minimal transition reward in $\calE$. Hence, a strategy which behaves like $\varrho$ and ``switches'' either to a strategy which mimics the optimal behaviour in $r(\ell)$ (when a configuration $r(k)$ for some $k \geq \ell$ is visited) or to some safe strategy (when a configuration with dangerously low counter is visited) is $\varepsilon$-optimal in every configuration $t(m)$ where $m \geq N_\varepsilon$. For configurations with smaller counter value, an $\varepsilon$-optimal startegy can be computed by transforming the configurations with a non-negative counter value bounded by $N_\varepsilon$ into a finite-state MDP and optimizing mean payoff in this finite-state MDP.

\paragraph{The case when $\calE$ is not strongly connected.}
We finish by considering the general case when $\calE$ is not strongly connected. Here, we again relay on standard methods for finite-state MDPs (see \cite{Puterman:book}). More precisely, we transform $\calE$ into a finite-state MDP $\M[\calE]$ in the following way. The states $\M[\calE]$ consist of those states of $\calE$ that do not appear in any MEC of $\calE$, and for each MEC $M$ of $\calE$ we further add a fresh controllable state $r_M$ to $\M[\calE]$. The transitions of $\M[\calE]$ are constructed as follows. For each $r_M$ we add a self-loop whose reward is the limit value of the states of the MEC~$M$ in $\calE$ (see the previous paragraph). Further, for every state $s$ of $\calE$, let $\hat{s}$ be either the state $s$ of $\M[\calE]$ or the state $r_M$ of $\M[\calE]$, depending on whether $s$ belongs to some MEC $M$ of $\calE$ or not, respectively. For every transition $(s,t)$ of $\calE$ where $s,t$ do \emph{not} belong to the same MEC, we add a transition $(\hat{s},\hat{t})$ to $\M[\calE]$. The rewards for all transitions, except for the self-loops on $r_M$, can be chosen arbitrarily.

Now we solve the standard mean-payoff optimization problem for $\M[\calE]$, which can be achieved in polynomial time by constructing a suitable linear program \cite{Puterman:book}. The program also computes a \emph{memoryless and deterministic} strategy $\sigma$ which achieves the optimal mean-payoff $\MP(s)$ in every state $s$  of $\M[\calE]$. Note that $\MP(r_M)$ is \emph{not} necessarily the same as the limit value of the states of $M$ computed by considering $M$ as a ``standalone EMDP'', because some other MEC with a better mean payoff can be reachable from~$M$. However, the strategy $\sigma$ eventually ``stays'' in some target $r_M$ almost surely, and the probability of executing a path of length $k$ before reaching a target $r_M$ decays exponentially in~$k$. Hence, for every $\delta > 0$, one can compute a bound $L_\delta$ such that the probability of reaching a target $r_M$ in at most  $L_\delta$ steps is at least $1 - \delta$. Moreover, $L_\delta$ is polynomial in $\size{\calE}$ and $1/\delta$. 

Now we show that $\MP(s) = \val(t)$ for every state $t$ of $\calE$ where $\hat{t} = s$. Further, we show that for every $\varepsilon \geq 0$, we can compute a sufficiently large $N_\varepsilon \in \Nset$ (still polynomial in $\size{\calE}$, $M_{\calE}$, and $1/\varepsilon$) and a strategy $\varrho$ such that for every initial configuration $t(m)$, where $m \geq N_\varepsilon$, we have that $\varrho$ is safe in $t(m)$ and $\E^{\varrho}_{t}[\MP] \geq \MP(s) - \varepsilon$, where $\hat{t} = s$. The strategy $\varrho$ ``mimics'' the strategy $\sigma$ and eventually switches to some other strategy (temporarily or forever) in the following way:
\begin{itemize}
	\item Whenever a configuration with a ``dangerously low'' counter value is encountered, $\varrho$ switches to a safe  strategy permanently. 
	\item In a controllable state $t$ of $\calM$ which does not belong to any MEC of $\calE$, $\varrho$ selects a transition $(t,u)$ such that $(t,\hat{u})$ is the transition selected by $\sigma$. In particular, if $\sigma$ selects a transition $(t,r_M)$, then $\varrho$ selects a transition leading from $t$ to some state of $M$.
	\item In a controllable state $t$ of a MEC $M$, $\varrho$ mimics $\sigma$ in the following sense. If $\sigma$ selects the transition $(r_M,r_M)$, then $\varrho$ permanently switches to the $\varepsilon/2$-optimal strategy for $M$ constructed in the previous paragraph. If $\sigma$ selects a different transition, then there must be a transition $(s,t)$ of $\calE$ where $s \in M$ such that $(r_M,\hat{t})$ is the transition selected by $\sigma$. Then $\varrho$ temporarily switches to a strategy which strives to reach the control state $s$. When $s$ is reached, $\varrho$ restarts mimicking $\sigma$. Note that for every \mbox{$\delta > 0$}, one can compute a bound $M_\delta$ polynomial in $\size{\calE}$ and $1/\delta$ such that the probability of reaching $s$ in at most $M_\delta$ steps is at least $1 - \delta$.
\end{itemize}
We choose $N_\varepsilon$ sufficiently large (with the help of the $L_\delta$ and $M_\delta$ introduced above) so that the probability of all runs initiated in $t(m)$, where $m \geq N_\varepsilon$, that reach a target MEC $M$ with a counter value above the threshold computed for $M$ and $\varepsilon/2$ by the methods of the previous paragraph, is at least $1 - \frac{\varepsilon}{2R}$, where $R$ is the difference between the maximal and the minimal transition reward in $\calE$. Hence, $\varrho$ is $\varepsilon$-optimal in every $t(m)$ where $m \geq N_\varepsilon$. For configuration with smaller initial counter value, we compute an $\varepsilon$-optimal
strategy as before.

Finally, let us note that Theorem~\ref{thm:onedim-general}~(5.) can be proven by reducing the following \emph{cost problem} which is known to be PSPACE-hard \cite{DBLP:conf/icalp/HaaseK15}: Given an acyclic MDP $\M=(\states,(\cstates,\stochstates),\trans{},\Prob,\rev)$, i.e., an~MDP whose graph does not contain an oriented cycle, a non-negative cost function $c$ (which assigns costs to transitions), an initial state $s_0$, a~target state $s_t$, a probability threshold $x$, and a bound $B$, decide whether there is a strategy which with probability at least $x$ visits $s_t$ in such a way that the total cost accumulated along the path is at most~$B$.
The reduction is straightforward and hence omitted.

%% file: app-hardness.tex
\newpage
\begin{center}
\LARGE \bf Technical Appendix
\end{center}
\section{Proofs}
\label{app-proofs}

In this section, we give full proofs that were omitted in the main body of the paper.

\begin{reftheorem}{Lemma}{\ref{lem:minpump}}
	For every EMDP $\calE$ there exists a memoryless strategy $\sigma$ such that $\sigma$ is pumping in every pumpable configuration of $\calE$. Further, there is a $\PTIME^{\EG}$~algorithm which computes the strategy $\sigma$ and the value $\minpump(s) \leq 3\cdot |\states|\cdot \maxE$ for every state $s$ of $\calE$. The problem whether a given configuration of $\calE$ is pumpable is $\EG$-hard.
\end{reftheorem}
\begin{proof}
We reduce the problem of computing $\minpump$ to the problem of computing minimal initial credit in \emph{energy parity MDPs}~\cite{CHD:energy-MDPs}, where we are required to find a safe strategy which visits with probability 1 a given set of states infinitely often. Given an EMDP $\calE$ we construct a new EMDP $\calE'$ by adding new states and transitions to $\calE$. For each transition $e=(s,t)$ of $\calE$ we add new controllable states $s_e,s'_e$ and transitions $(s,s_e),(s_e,t)$, $(s_e,s'_e)$, $(s'_e,s_e)$ such that $E(s_e,s'_e)=-1$ and the other three transitions have energy update 0 (the reward of the new transitions is irrelevant). We require that some state if the form $s'_e$ is visited infinitely often, i.e. that the counter is infinitely often decreased by 1. It is easy to verify that a configuration is pumpable if and only if it admits a safe strategy that satisfies this B\"uchi objective with probability one. 

To determine minimal initial energy level needed to achieve the latter, in~\cite{CHD:energy-MDPs} the authors provide a polynomial reduction to determining the minimal initial level in energy B\"uchi games, a problem which is shown to be solvable by an $\PTIME^{\EG}$ algorithm in \cite{CD:energy-parity-games}. For memorylessness, assume that $\calE$ is pumpable and let $\calE''$ be an EMDP obtained by removing all transitions $(s,t)$ such that $\minpump(s)+\energy(s,t)<\minpump(t)$, and removing all states $s$ for which $\minpump(s)=\infty$. It is easy to check that $\minpump$-values of states in $\calE''$ are the same as in $\calE$, and moreover, \emph{any} strategy in $\calE''$ is safe in all safe configurations, so in particular there are no negative cycles in $\calE''$. Moreover, in $\calE''$, it must be possible to reach, from each state, a positive cycle with probability 1, otherwise the said state would be unpumpable with any initial energy level. Hence, we can pick a set $\mathit{\Pi}$ of disjoint positive cycles such that at least one cycle in $\Pi$ is reachable from each state of $\calE''$ a define a memoryless strategy $\pi$ in such a way that in a state on one of these cycles it selects a transition (of $\calE''$) which keeps us on the cycle and in all other states it selects a transition which takes us closer to some of these cycles (optimal strategies for reachability are memoryless). It is then easy to show that $\pi$ is a globally pumping strategy in $\calE''$ and thus also in $\calE$.
\qed
\end{proof}

\subsection{Proofs of Section~\ref{sec-SPEMDPs}}

Recall that we assume a fixed strongly connected and pumpable EMDP $\calE=(\M,\energy)$ where $\M = (\states,(\cstates,\stochstates),\trans{},\Prob,\rev)$. Let $\vec{f}$ be an optimal solution to the program $\calL_{\calE}$ of Figure~\ref{fig:oc-lp} with optimal value $f^*$.

We start by considering the case where we compute a type I core of $\vec{f}$, i.e. on the proof of Proposition~\ref{prop:typeA-optimality}. %

\subsubsection{Proof of Proposition~\ref{prop:typeA-optimality}}

Let $C$ be a type I core of $\vec{f}$, $s(n)$ a configuration of $\calE$, and let strategy $\sigma_n^*$ be as in Proposition~\ref{prop:typeA-optimality}. If $s(n)$ is not safe, then $\sigma_n^*$ any strategy is optimal in $s(n)$, so assume that $s(n)$ is safe. We prove that $\sigma_n^*$ is optimal in $s(n)$. First note that $\sigma_n^*$ is clearly safe in $s(n)$, since whenever we are configuration $t(\ell)$ with $\ell \leq \maxE + \minpump(t)$, the strategy $\mu_C$ starts to behave as a globally pumping strategy which never visits a configuration $t'(\ell')$ with $\ell'\leq \minpump(t')$, and moreover, such $t'(\ell')$ cannot be visited without previously visiting a configuration $t''(\ell'')$ with $\minpump(t'')\leq \ell''\leq \minpump(t'')+\maxE$. So we focus on optimality of the mean payoff produced by $\sigma_n^*$,

First note that the memoryless strategy $\mu_C$, one of the two constituent strategies of $\sigma_n^*$, achieves mean payoff $f^*$ from each state of $\calE$~\cite[Lemma 4.3]{BBCFK14}, and the long-run change of the energy level under $\mu_C$ is positive. In particular, it suffices to prove that with probability 1 the strategy $\sigma_n^*$ eventually starts to behave as $\mu_C$ and sticks to this behaviour \emph{forever}, or formally, that under $\sigma_n^*$ it holds with probability one that for all but finitely many prefixes of $w$ of the produced run we have $\mathit{low}_n(w)=0$. To show this, we use the following fact:

\begin{lemma}
The following holds for all $t\in \states$ and $m\geq H$: For every state $t$, starting in configuration $t(m)$ with strategy $\mu_C$, the probability that we eventually encounter a configuration $t'(m')$ with $m'\leq L$ is strictly smaller than $1$. 
\end{lemma}
\begin{proof}
We first present the proof under the assumption that $C=S$ and $\maxE\leq 1$.

Since $C$ has positive trend, the expected long-run change  of the counter under $\mu_C$ is positive. 
From~\cite[Lemma 4]{BKKNK:multicounter-zero-reachability} it follows that the probability of never hitting energy level $\leq  L$ is positive for each initial energy level $m$ greater than \emph{some} finite bound $H'\geq L$. We prove that this finite bound can be assumed to be $L+|\states|\leq H$.

For any $i\geq L+1$ denote by $\mathcal{Z}_i$ the set of all states $s$ of $\calE$ such that under strategy $\mu_C$ the probability of the energy level decreasing to $L$ when starting in $s(i)$ equals $1$. Note that $s\in \mathcal{Z}_i$ if and only if the following two conditions hold:
\begin{itemize}
\item
When starting in $s(i)$ with strategy $\mu_C$, the probability of decreasing the energy level to $i-1$ is $1$.
\item
Denoting by $\mathcal{R}_i$ the set of all states $t$ such that configuration $t(i-1)$ is encountered with positive probability when starting in $s(i)$ with $\mu_C$, it holds $\mathcal{R}_i \subseteq \mathcal{Z}_{i-1}$.
\end{itemize}

Note that if condition (1.) holds for at least one configuration of the form $s(i)$, it holds for all $s(i)$ s.t. $i\geq  L$, since strategy $\mu_C$ is memoryless. As noted above, it holds for $s(H')$, so it holds for all $s(i)$ with $i\geq L$. Whether the second condition holds for $s(i)$ depends solely on $\mathcal{Z}_{i-1}$, as $\mathcal{R}_i = \mathcal{R}_{i'}$ for all $i$, $i'$, again due to memorylessness of $\mu_C$. Hence, if $\mathcal{Z}_{i}=\mathcal{Z}_{i+1}$, then $\mathcal{Z}_{i}=\mathcal{Z}_{i'}$ for all $i'\geq i$. Moreover, $\mathcal{Z}_i \supseteq \mathcal{Z}_{i+1}$ for all $i$, since if memoryless strategy $\mu_C$ almost surely decreases the energy level to $ L$ from some $u(i+1)$, it does the same from $u(i)$ as well. Hence, it must be the case that $\mathcal{Z}_{L+|\states|}=\mathcal{Z}_{L+|\states|+1}$ and thus $\mathcal{Z}_{i}=\mathcal{Z}_{H'}$ for all $i\geq L+|\states|$. As shown above, $\mathcal{Z}_{H'}=\emptyset$, which finishes the proof for the special case.

Now we drop the assumption that $\maxE\leq 1$. We can then subdivide each transition $(s,t)$ with $\energy(s,t)=e$ into a path of length $\maxE$ on which each edge is labelled by $e/\maxE$ (assignment of rewards is irrelevant). Thus, we reduce the proof to the case with $\maxE$ at the cost of blowing-up the state space: the transformed EMDP $\calE'$ has at most $|\states|^2\cdot \maxE + |\states|$ states. The strategy $\mu_C$ can be straightforwardly carried over to this EMDP, and it is easy to check that the expected long-run change of the counter under $\mu_C$ is the same in $\calE$ and $\calE'$, in particular it is positive. Moreover, for each state $t$ of the original MDP its $\minpump$-value is the same in both EMDPs. We can thus apply the results of the previous paragraph to $\calE'$ and get that the probability of hitting energy level $L$ from $s(i)$ using $\mu_C$ is less than 1 for each $i\geq L + 2|\states|^2 \cdot\maxE$.

It remains to lift the assumption that $C=\states$. So let $C\subset\states$. Since $\mu_C$ reaches $C$ almost surely from each state $s\in S$, and $\mu_C$ is memoryless, we know that from each such state $s$ there is a path $w$ of length at most $|S|$ such that $w$ ends within $C$ and is traversed with positive probability. So starting in configuration $s(L + |\states| + 2|\states|^2\cdot\maxE)=s(H)$ and using strategy $\mu_C$, we are guaranteed that with positive probability we hit a configuration $t(\ell)$ with $\ell\geq L + 2|\states|^2\cdot\maxE$ and $t\in C$ without hitting a configuration with energy level smaller than $L$. By previous paragraph, from $t(\ell)$ we have a positive probability of never going below $L$, which finishes the proof.

\end{proof}

Now we finish the proof of Proposition~\ref{prop:typeA-optimality}. Suppose that with positive probability we infinitely often encounter the situation when the function $\mathit{low}_n$ attains value 1. After each such occasion the strategy $\sigma$ eventually switches back to behaving as $\mu_C$, since $\pi$ is a globally pumping strategy. When this switch occurs, there is a positive probability (bounded away from zero) that we will never encounter the situation with $\mathit{low}_n=0$ again, as shown by the previous lemma. It follows, that the probability of infinitely often seeing such a situation is zero, a contradiction.

\subsubsection{Proof of Proposition~\ref{prop:type-II-optimality}}

To define an optimal strategy $\sigma_n^*$, we need additional notation: For $w=s_0s_1\cdots$ and $0\leq i\leq \len{w}$ we denote by $\pathstate{w}{i}$ the state $s_i$.

We first prove a couple of useful general lemmas.

In the following we mean by ``playing according to a memoryless strategy $\mu$''  that at each situation we select a distribution on actions prescribed by $\mu$ for the current state. We also use this terminology for history-dependent strategies: when saying that at some point (after observing a history $w$) we ``play according to some strategy $\sigma$,'' we mean that from this point on, after seeing a history $ww'$ we choose the distribution on actions given by $\sigma(w')$.

\begin{lemma}
\label{lem:reach-mixing}
Let $\mu_1$, $\mu_2$ be memoryless strategies in $\calE$, $p_1,p_2 \in [0,1]$ numbers s.t. $p_1 + p_2 =1$, $K\in \Nset$, $N\in \Nset$  the smallest number s.t. $p_1\cdot N$ and $p_2\cdot N$ are integers, and let $q$ be any state of $\calE$. Assume that both $\mu_1$ and $\mu_2$ determine a Markov chain with a single bottom strongly connected component (i.e. using $\mu_i$, almost all runs have the same frequency of visits to a given state).

For each $i\in \Nset$ let $T_i$ be a probability distribution on $\Nset_0$ for which there exist a~function $g:\Nset\rightarrow \Nset$ and a constant $c\in (0,1)$ satisfying $\Probm(T_i\geq g(i))\leq c^{-i}$ and $\lim_{i\rightarrow \infty} \sum_{i=1}^n g(i)/n^2=0$.

Finally, let $\sigma$ be a strategy in $\calE$ defined as follows: $\sigma$ is played in stages. In stage $i\in \Nset$, we:
\begin{itemize}
\item  First play according to $\mu_1$ for exactly $p_1\cdot N\cdot i$ steps,
\item then play according to $\mu_2$ for exactly $p_2\cdot N\cdot i$ steps,
\item then play according to a memoryless deterministic strategy $\kappa$ which guarantees reaching $q$ with probability~1 (such a strategy exists due to $\calE$ being strongly connected). We play according to $\kappa$ until $q$ is reached. 
\item Then, play according to a globally pumping strategy $\pi$ (which is guaranteed to exist by Lemma~\ref{lem:minpump}). We play according to $\pi$ for a random number of steps determined by a single draw from the distribution $T_i$.
\item
Then we proceed to stage $i+1$.
\end{itemize} 

\noindent
Then for all states $s$ it holds $\E^{\sigma}_s [\MP] = p_1\cdot \E^{\mu_1}_s[\MP] + p_2\cdot \E^{\mu_2}_s[\MP]$.
\end{lemma}
\begin{proof}
Let us denote by $M^{\mu_1}_i$, $M^{\mu_2}_i$, $M^{\kappa}_i$ and $M^{\pi}_i$ the total rewards accumulated during the $i$-th stage playing according to $\mu_1$, $\mu_2$, $\kappa$ and $\pi$. Denote by $L^{\kappa}_i$ the number of steps made according to $\kappa$ in the $i$-th stage. Slightly abusing notation, we use $T_i$ to denote the~number of steps made according to $\pi$ in the $i$-th stage, and assume that $T_1,T_2,\ldots$ are independent. Denote by $\bar{L}_i$ the length of the $i$-the stage, i.e. $N\cdot i+L^{\kappa}_i+T_i$. 

We use the~following equation (which will be justified below): Almost surely,
\begin{equation}\label{eq:MP_decomp}
\MP = 
\lim_{n\rightarrow\infty}\frac{\sum_{i=1}^n M^{\mu_1}_i+M^{\mu_2}_i+M^{\kappa}_i+M^{\pi}_i}{\sum_{i=1}^n \bar{L}_i}=p_1\cdot \E^{\mu_1}_s[\MP]+p_2\cdot \E^{\mu_2}_s[\MP]
\end{equation}
First, we show 
\begin{equation}\label{eq:MP_result}
\lim_{n\rightarrow\infty}\frac{\sum_{i=1}^n M^{\mu_1}_i+M^{\mu_2}_i+M^{\kappa}_i+M^{\pi}_i}{\sum_{i=1}^n \bar{L}_i}=p_1\cdot \E^{\mu_1}_s[\MP]+p_2\cdot \E^{\mu_2}_s[\MP]
\end{equation}
Then we finish the proof by proving (\ref{eq:MP_decomp}).
We have
\begin{align}
\lim_{n\rightarrow\infty} & \frac{\sum_{i=1}^n M^{\mu_1}_i+M^{\mu_2}_i+M^{\kappa}_i+M^{\pi}_i}{\sum_{i=1}^n \bar{L}_i}=\label{eq:mp_decomp2}\\
& \lim_{n\rightarrow\infty}\frac{\sum_{i=1}^n M^{\mu_1}_i+M^{\mu_2}_i}{\sum_{i=1}^n N\cdot i}\frac{\sum_{i=1}^n N\cdot i}{\sum_{i=1}^n \bar{L}_i}+\lim_{n\rightarrow\infty}\frac{\sum_{i=1}^n M^{\kappa}_i}{\sum_{i=1}^n \bar{L}_i}+\lim_{n\rightarrow\infty}\frac{\sum_{i=1}^n M^{\pi}_i}{\sum_{i=1}^n \bar{L}_i}
\end{align}
assuming that the limits on the right-hand side exist.

One can easily show that, a.s., 
\begin{eqnarray*}
\lim_{n\rightarrow\infty}\frac{\sum_{i=1}^n M^{\mu_1}_i+M^{\mu_2}_i}{\sum_{i=1}^n N\cdot i} & = & \lim_{n\rightarrow\infty}\frac{\sum_{i=1}^n M^{\mu_1}_i}{\sum_{i=1}^n p_1\cdot N\cdot i}\lim_{n\rightarrow\infty}\frac{\sum_{i=1}^n p_1\cdot N\cdot i}{\sum_{i=1}^n N\cdot i}\\
& \qquad + &
\lim_{n\rightarrow\infty}\frac{\sum_{i=1}^n M^{\mu_2}_i}{\sum_{i=1}^n p_2\cdot N\cdot i}\lim_{n\rightarrow\infty}\frac{\sum_{i=1}^n p_2\cdot N\cdot i}{\sum_{i=1}^n N\cdot i}\\
& = & p_1\cdot \E^{\mu_1}_s[\MP]+p_2\cdot \E^{\mu_2}_s[\MP]
\end{eqnarray*}
Here the last equality follows from the ergodic theorem for finite-state Markov chains (see~e.g.~\citeaddr{Norris:book}) applied to $\mu_1$ and to $\mu_2$.

So to prove (\ref{eq:MP_result}) it suffices to prove the following equations (and apply (\ref{eq:mp_decomp2})):
\begin{equation}\label{eq:A}
\lim_{n\rightarrow\infty}\frac{\sum_{i=1}^n \bar{L}_i}{\sum_{i=1}^n N\cdot i}=1
\end{equation}
\begin{equation}\label{eq:B}
\lim_{n\rightarrow\infty}\frac{\sum_{i=1}^n M^{\kappa}_i}{\sum_{i=1}^n\bar{L}_i}=0
\end{equation}
\begin{equation}\label{eq:C}
\lim_{n\rightarrow\infty}\frac{\sum_{i=1}^n M^{\pi}_i}{\sum_{i=1}^n\bar{L}_i}=0
\end{equation}
We start by proving two auxiliary claims:
\begin{claim}[1]
\[
\lim_{n\rightarrow\infty}\frac{\sum_{i=1}^n L^{\kappa}_i}{n}<\infty
\]
\end{claim}
\begin{proof}[of the claim]
let us define $L^{\kappa}_{i,s'}$ the number of steps played according to $\kappa$ in the $j$-th stage where $\kappa$ starts in $s'$. Given $n$ denote by $n_{s'}$ the number of such stages up to the $n$-th stage. Then for every $s'$ the $L^{\kappa}_{1,s'}, L^{\kappa}_{2,s'},\ldots$ are independent and identically distributed with $\E^{\sigma}_s(L^{\kappa}_{j,s'})=\E^{\sigma}_s(L^{\kappa}_{1,s'})<\infty$, and hence by invoking the strong law of large numbers for iid variables (see~e.g.~\cite{Williams:book}) we obtain
\begin{align*}
\lim_{n\rightarrow\infty} & \frac{\sum_{i=1}^n L^{\kappa}_i}{n}  = 
\lim_{n\rightarrow\infty}\frac{\sum_{s'} \sum_{j=1}^{n_{s'}} L^{\kappa}_{j,s'}}{n}=\sum_{s'}\lim_{n\rightarrow\infty}\frac{\sum_{j=1}^{n_{s'}} L^{\kappa}_{j,s'}}{n}=\\
&\sum_{s'}\lim_{n\rightarrow\infty}\frac{\sum_{j=1}^{n_{s'}} L^{\kappa}_{j,s'}}{n_{s'}}\lim_{n\rightarrow\infty} \frac{n_{s'}}{n} \leq  \max_{s'} \lim_{n\rightarrow\infty} \frac{\sum_{j=1}^{n_{s'}} L^{\kappa}_{j,s'}}{n_{s'}} = \max_{s'} \E^{\sigma}_s(L^{\kappa}_{j,s'})<\infty
\end{align*}
This finishes the proof of Claim~(1).
\end{proof}

\begin{claim}[2]
\[
\lim_{n\rightarrow\infty}\frac{\sum_{i=1}^n T_i}{\sum_{i=1}^n i}=0
\]
\end{claim}
\begin{proof}[of the claim]
By our assumptions, $\Probm(T_i\geq g(i))\leq c^{-i}$ for all $i$ and thus $\sum_{i=1}^{\infty} \Probm(T_i\geq g(i))<\infty$. Hence, by Borel-Cantelli lemma  (see~\cite{Williams:book}), for almost every run there is $i'$ such that $T_i<g(i)$ for $i\geq i'$. However, then, a.s.,
\[
\lim_{n\rightarrow \infty} \frac{\sum_{i=1}^n T_i}{\sum_{i=1}^n i}=
\lim_{n\rightarrow \infty} \frac{\sum_{i=i'}^n T_i}{\sum_{i=i'}^n i}< 
\lim_{n\rightarrow \infty} \frac{\sum_{i=i'}^n g(i)}{\sum_{i=i'}^n i}=
\lim_{n\rightarrow \infty} \frac{\sum_{i=i'}^n g(i)}{\frac{n}{2}(n+1)}=0
\]
Here the last equality follows from our assumptions on $g$.
This finishes the proof of the claim~(2).
\end{proof}

Let us prove the equation (\ref{eq:A}).
\begin{eqnarray*}
\lim_{n\rightarrow\infty}\frac{\sum_{i=1}^n \bar{L}_i}{\sum_{i=1}^n N\cdot i} & = & 
\lim_{n\rightarrow\infty}\frac{\sum_{i=1}^n N\cdot i+L^{\kappa}_i+T_i}{\sum_{i=1}^n N\cdot i}\\
& = & \lim_{n\rightarrow\infty}\frac{\sum_{i=1}^n N\cdot i}{\sum_{i=1}^n N\cdot i}+\lim_{n\rightarrow\infty}\frac{\sum_{i=1}^n L^{\kappa}_i}{\sum_{i=1}^n N\cdot i} + \lim_{n\rightarrow\infty}\frac{\sum_{i=1}^n T_i}{\sum_{i=1}^n N\cdot i}\\
& = & 1+\lim_{n\rightarrow\infty}\frac{\sum_{i=1}^n L^{\kappa}_i}{n}\lim_{n\rightarrow\infty}\frac{n}{\sum_{i=1}^n N\cdot i} + \lim_{n\rightarrow\infty}\frac{\sum_{i=1}^n T_i}{\sum_{i=1}^n N\cdot i}\\
& = & 1\\
\end{eqnarray*}
The last equality follows from Claim~(1) and Claim~(2).
This finishes the proof of (\ref{eq:A}). 

Now let us prove (\ref{eq:B}):
\begin{eqnarray*}
\lim_{n\rightarrow\infty}\frac{\sum_{i=1}^n M^{\kappa}_i}{\sum_{i=1}^n\bar{L}_i} & \leq  & \lim_{n\rightarrow\infty}\frac{\sum_{i=1}^n L^{\kappa}_i\cdot \max r}{\sum_{i=1}^n\bar{L}_i}\\
& =  &
\max r\cdot \lim_{n\rightarrow\infty}\frac{\sum_{i=1}^n L^{\kappa}_i}{n}\cdot\lim_{n\rightarrow \infty} \frac{n}{\sum_{i=1}^n N\cdot i}\cdot \lim_{n\rightarrow \infty} \frac{\sum_{i=1}^n N\cdot i}{\sum_{i=1}^n\bar{L}_i}\\
& = & 0
\end{eqnarray*}
Here the last equality follows from Claim~(1) and the equation~(\ref{eq:A}).
Similarly, using Claim~(2), we prove (\ref{eq:C}):
\begin{eqnarray*}
\lim_{n\rightarrow\infty}\frac{\sum_{i=1}^n M^{\pi}_i}{\sum_{i=1}^n\bar{L}_i} & \leq  & \lim_{n\rightarrow\infty}\frac{\sum_{i=1}^n T_i\cdot \max r}{\sum_{i=1}^n\bar{L}_i}\\
& =  &
\max r\cdot \lim_{n\rightarrow\infty}\frac{\sum_{i=1}^n T_i}{\sum_{i=1}^n N\cdot i}\cdot \lim_{n\rightarrow \infty} \frac{\sum_{i=1}^n N\cdot i}{\sum_{i=1}^n\bar{L}_i}\\
& = & 0
\end{eqnarray*}
To finish the proof of Lemma~\ref{lem:reach-mixing} we prove that $\MP$ exists a.s. Then (\ref{eq:MP_decomp}) follows from (\ref{eq:MP_result}) and the fact that the sequence on the right-hand side of (\ref{eq:MP_decomp}) is a subsequence of the mean-payoff defining sequence. 
Denote by $\MP_j$ the $j$-the average of the rewards obtained in the first $j$ steps. Denote by $k_j$ the number of stages completed in the first $j$ steps.

Observe that
\[
\frac{\sum_{i=1}^{k_j} M^{\mu_1}_i+M^{\mu_2}_i+M^{\kappa}_i+M^{\pi}_i}{\sum_{i=1}^{k_j} \bar{L}_i} \leq \MP_j \leq \frac{\sum_{i=1}^{k_j} M^{\mu_1}_i+M^{\mu_2}_i+M^{\kappa}_i+M^{\pi}_i+\bar{L}_{k_j+1}\cdot\max r}{\sum_{i=1}^{k_j} \bar{L}_i}
\]
Note that limits of the left-hand side and the right-hand side are equal as $j$ goes to infinity, and of course, $\lim_{j\rightarrow\infty} \MP_j=\MP$. Indeed, observe
\begin{align*}
\lim_{m\rightarrow \infty} & \frac{\bar{L}_{m+1}}{\sum_{i=1}^{m} \bar{L}_i}\\
& =
\lim_{m\rightarrow \infty} \frac{N\cdot (m+1)+L^{\kappa}_{m+1}+T_{m+1}}{\sum_{i=1}^{m} N\cdot i+L^{\kappa}_i+T_i}\\
& \leq 
\lim_{m\rightarrow \infty} \frac{N\cdot (m+1)+L^{\kappa}_{m+1}+T_{m+1}}{\sum_{i=1}^{m} i}\\
& \lim_{m\rightarrow \infty} \frac{N\cdot (m+1)+L^{\kappa}_{m+1}+T_{m+1}}{\sum_{i=1}^{m+1} i}\frac{\sum_{i=1}^{m+1} i}{\sum_{i=1}^{m} i}\\
&=0
\end{align*}
Here the last equality follows from Claim~(1), Claim~(2) and the fact that 
$\lim_{m\rightarrow\infty} \frac{\sum_{i=1}^{m+1} i}{\sum_{i=1}^{m} i}=1$.

This finishes the proof of Lemma~\ref{lem:reach-mixing}.

\end{proof}

Now let $C_1,C_2$ be a type II core of $\vec{f}$, and $s(n)$ a configuration of $\calE$. We again assume that $s(n)$ is safe. 

As in the type I case, the components $C_1$, $C_2$ induces memoryless strategies $\mu_1$, $\mu_2$ such that for each $i\in\{1,2\}$ the strategy $\mu_i$ behaves as follows: inside $C_i$ it plays according to frequencies obtained from $\vec{f}$ and outside of $C_i$ it behaves as a memoryless deterministic strategy for reaching $C_i$ with probability 1. Note that both $\mu_i$ induce a Markov chain with a single bottom strongly connected component. 

Let $p_1 = f_{C_1}$ and $p_2 = f_{C_2}$, $N\in \Nset$  the smallest number s.t. $p_1\cdot N$ and $p_2\cdot N$ are integers, and let $q$ be an arbitrary state of $\calE$. 
We define a strategy $\sigma$ as follows: $\sigma_1$ is executed in stages. In stage $i\in \Nset$, we:
\begin{itemize}
\item First play according to $\mu_1$ for exactly $p_1\cdot N\cdot i$ steps,
\item then play according to $\mu_2$ for exactly $p_2\cdot N\cdot i$ steps,
\item then play according to a memoryless deterministic strategy $\kappa$ which guarantees reaching $q$ with probability~1 (such a strategy exists due to $\calE$ being strongly connected). We play according to $\kappa$ until $q$ is reached. 
\item Then, play according to a globally pumping strategy $\pi$ (which is guaranteed to exist by Lemma~\ref{lem:minpump}). We play according to $\pi$ until the energy level is at  least $\mathit{TH} + (i\cdot N)^{\frac{3}{4}} $, where  $\mathit{TH}=\max_{q\in \states}\minpump(q)+\maxE$.
\item
Then we proceed to stage $i+1$.
\end{itemize}

\noindent
Note that strategy $\sigma$ \emph{is not} safe in general.

\begin{lemma}
\label{lem:mixing-overall}
Strategy $\sigma_1$ satisfies $\E^{\sigma}_s [\MP] = p_1\cdot \E^{\mu_1}_s[\MP] + p_2\cdot \E^{\mu_2}_s[\MP].$ In particular, $\E^{\sigma}_s [\MP]=f^*$.
\end{lemma}
\begin{proof}
We use Lemma~\ref{lem:reach-mixing}. The only thing we need to prove is to show that in each segment $i$, the random variable $T_i$ denoting the time for which we play the globally pumping strategy $\pi$ satisfies the condition in the assumptions of Lemma~\ref{lem:reach-mixing}. That is, we need to find the right function $g$ and constant $c$.

Note that in each stage we start playing according to $\pi$ while in a state $q$. Memoryless strategy $\pi$ induces a finite Markov chain $M_{\pi}$ whose states are exactly the states of $\calE$. Let $C_1,\dots,C_\ell$ be all the bottom strongly connected components (BSCCs) of $M_{\pi}$ that are reachable from $q$ in $M_\pi$. It is easy to check that to satisfy the assumptions of Lemma~\ref{lem:reach-mixing} we need to prove the following:

\begin{itemize}
\item Denoting by $T^1$ the number of steps elapsed until one of the BSCCs $C_1,\dots,C_\ell$ is reached, there exist a~function $g_1:\Nset\rightarrow \Nset$ and a constant $c_1\in (0,1)$ satisfying $\Probm^{\pi}_{q}(T^1 \geq g_1(i))\leq c_1^{-i}$ and $\lim_{i\rightarrow \infty} \sum_{i=1}^n g_1(i)/n^2=0$ for all $i$.
\item For all states $t$ that belong to one of the components $C_1,\dots,C_\ell$, there exist a~function $g_2:\Nset\rightarrow \Nset$ and a constant $c_2\in (0,1)$ satisfying $\Probm^{\pi}_{t}(T \geq g_2(i))\leq c_2^{-i}$ and $\lim_{i\rightarrow \infty} \sum_{i=1}^n g_2(i)/n^2=0$ for all $i$.
\end{itemize}

The existence of $g_1$ and $c_1$ is easy, it follows, e.g. from \cite[Lemma 5.1]{BKK:oc-jacm}. 

Now fix any state $t$ as prescribed above. Note that from the construction of $\pi$ it follows that it's counter trend $\trend_\pi$ from $q$ (i.e. the number $\E^{\pi}_q \lim_{k\rightarrow \infty}  \sum_{i=1}^{k}e_i(\omega)/k$, where $e_i(\omega)$ is the energy change on the $i$-th transition of $\omega$) is positive (see the proof of Lemma~\ref{lem:minpump} -- all cycles visited by the strategy have non-negative effect, and with probability 1 we infinitely often traverse a cycle of positive effect. Since $\pi$ is memoryless, the probability of large gaps between two traversals of a positive cycle decays exponentially with the size of the gap, from which the result follows via standard computations). Since $t$ is in a BSCC of the Markov chain induced by $\pi$, from \cite{BKK:oc-jacm} it follows that under $\pi$ there is a \emph{bounded-difference martingale}, a stochastic process $(\ttmart{j})_{j=0}^{\infty}$ given by $\ttmart{j}(\omega)=\enlev{\ell}{\omega}{j} + \bar{z}(\pathstate{\omega}{j}) - j \cdot \trend_{\pi}$ for some weight function $\bar{z}\colon \states \rightarrow \Rset$, where $\ell$ is the energy level in which we enter the BSCC in $t$. 

Now any run $\omega$ initiated in $t$ along which the energy level does not increase above $\mathit{TH} + (i\cdot N)^{\frac{3}{4}} $ in the first $W_i=(2\cdot N \cdot i+\mathit{TH})/\trend_\pi$ steps satisfies $|\ttmart{W_i}(\omega)-\ttmart{0}(\omega)|\geq i -2Z$, where $Z=\max_{t'}\bar{z}(t')$. From the Azuma's inequality~\cite{Williams:book} it follows that for all but finitely many $i$ the probability $\Probm^{\pi}_{t}(T\geq W_i) $ is bounded from above by $ c_2^{i}$ for a suitable number $c_2\in(0,1)$. Hence, it suffices to put $g_2(i)=W_i$ for all such $i$. For the finitely many remaining $i$'s we can set $g_2(i)$ to any number $W$ such that the maximum among all these finitely many $i$'s of the probability $\Probm^{\pi}_{t}(T\geq W) $ is smaller than, say $\frac{1}{2}$ (such a $W$ exists, since $\pi$ is pumping).

\end{proof}

\newcommand{\sigmaopt}{{\sigma_n^*}}

Now we modify $\sigma$ to make it safe: in each stage, we play as prescribed above. However, if the current energy level falls below the threshold $\mathit{TH}=\max_{q\in \states}\minpump(q)+\maxE$, we immediately skip to the second-to-last item, i.e. to the use of the globally pumping strategy $\pi$, which is played until the energy level surpasses the value prescribed for the current stage ($(i\cdot N)^{\frac{3}{4}} $). Denote this strategy $\sigma_n^*$. It is clear that $\sigma_n^*$ is safe (it is actually pumping as well). It remains to prove that $\E^{\sigma_n^*}_{s}[\MP]=\E^{\sigma}_{s}[\MP]$, i.e. that $\sigma_n^*$ is optimal.

We say that a stage $i$ of $\sigma_n^*$ \emph{fails} if the energy level falls below $\mathit{TH}$ during this stage. To prove that $\sigmaopt$ is optimal it suffices to prove that with probability 1, only finitely many stages of $\sigmaopt$ fail (and thus $\sigmaopt$ eventually starts to behave as $\sigma$ forever). Due to Borel-Cantelli lemma it suffices to show that $\sum_{i=1}^{\infty}\Probm^{\sigmaopt}_s(\text{$\sigmaopt$ fails in stage $i$})<\infty$. We prove that there is $i\in(0,1)$ such that for all but finitely many $i$'s the probability of failure in stage $i$ is bounded by $c^{\sqrt{i}}$, which yields a converging infinite sum.

So let $i$ be arbitrary and let $t$ be an arbitrary state in which stage $i$ starts. Note that stage $i$ starts with energy level at least $\mathit{TH}+L_i$, where $L_i=(i\cdot N)^{\frac{3}{4}}$.

\newcommand{\failreach}{\mathit{FR}}
\newcommand{\failincrease}{\mathit{FI}}
\newcommand{\faildecrease}{\mathit{FD}}
Consider the following events that may happen in stage $i$:
\begin{enumerate}
\item $F_1$: When starting in $t$, it takes at least $\frac{1}{6}L_i$ steps to reach $C_1$.
\item $F_2$: $\neg F_1$ and inside $C_1$ the counter increases by less than $f_{C_1}\cdot N\cdot i \cdot \trend_{C_1} - \frac{2}{6}L_i$ before we start to play according to $\mu_2$.
\item $F_3$:  $\neg F_1$ and $\neg F_2$ and inside $C_1$ the counter decreases below $\mathit{TH}$ before we start to play according to $\mu_2$.
\item $F_4$: $\bigcap_{j=1}^3 \neg F_j$ and upon starting to play according to $\mu_2$, it takes at least $\frac{1}{6}L_i$ steps to reach $C_2$.
\item $F_5$: $\bigcap_{j=1}^4 \neg F_j$ and inside $C_2$ the counter decreases by more than $f_{C_2}\cdot N\cdot i \cdot \trend_{C_2} + \frac{1}{6}L_i$ before we start to play according to $\kappa$.
\item $F_6$: $\bigcap_{j=1}^5 \neg F_j$ and upon starting to play according to $\kappa$, it takes at least $\frac{1}{6}L_i$ steps to reach $q$.
\end{enumerate}

Note that if \emph{none} of the events happens during the $i$-th stage, then this stage \emph{does not} fail. Of particular interest here is the event $F_5$: note that if $\bigcap_{j=1}^4 \neg F_j$ happens, then when we enter $C_2$ while playing according to $\mu_2$, our energy level is at least $L_i+ f_{C_1}\cdot N\cdot i \cdot \trend_{C_1} - \frac{4}{6}L_i$, so if $F_5$ holds, upon starting the play according to $\kappa$ our energy level is at least $\mathit{TH}+L_i+ f_{C_1}\cdot \trend_{C_1}\cdot N\cdot i  + f_{C_2}\cdot \trend_{C_2}\cdot N\cdot i  - \frac{5}{6}L_i = \mathit{TH}+\frac{1}{6}L_i$ (we have $f_{C_1}\cdot \trend_{C_1} + f_{C_2}\cdot \trend_{C_2}=0$, since $C_1,C_2$ is a type II core of $\vec{f}$). Now to find $c$ whose existence is postulated above, it is sufficient to find, for each of the above events, a number $d\in (0,1)$ such that for all but finitely many $i$'s the probability of the said event is bounded by $d^{\sqrt{i}}$.

For events $F_1$, $F_4$, and $F_6$, we can again invoke  Lemma~5.1. of~\cite{BKK:oc-jacm}. The lemma proves that in a finite Markov chain (such as the one induced by a memoryless strategy for reaching some set of states) we can find a number $d'\in(0,1)$ such that the probability of not reaching a given almost-surely reachable set within $\ell$ steps is at most ${d'}^{\ell}$. In our cases we have $\ell=i^{\frac{3}{4}}\cdot b$, where $b$ is independent of $i$, which proves the existence of $d$.

For the remaining events we need to use arguments based on \emph{martingales}~\cite{Williams:book}. Let us start with $F_3$. From Theorem~3.4. of~\cite{BKK:oc-jacm} it follows that there is a \emph{weight} function $z\colon\states\rightarrow \Qset$ such that for any $n\in \Zset$ following stochastic process $(\mart{i})_{i = 0}^{\infty}$ is a martingale under $\mu_{C_1}$ when starting in $C_1$:\footnote{Although~\cite{BKK:oc-jacm} considers only a special case when $\maxE=1$, the proof works also for our model without any modification.} %
	\[
	\mart{i}(\omega) =
	\enlev{n}{\omega}{i} + z(\pathstate{\omega}{i}) - i \cdot \trend_{C_1}.
	\]
	
	Moreover, from standard results on martingales, we get that if we denote by $\tau(\omega)$ the first point in time in which the energy level drops below $\mathit{TH}$, then the process $(\maralt{i})_{i=0}^{\infty}$, where $\maralt{i}(\omega)=\mart{\min\{i,\tau(\omega)\}}(\omega)$, is also a martingale. Moreover, both martingales have \emph{bounded differences}, i.e. their one-step change is bounded uniformly over all runs and steps. Now any run $\omega$ initiated in some $u(\ell)$, $u\in C_1$, $\ell\geq\mathit{TH} + \frac{5}{6}L_i$ whose energy level drops below $\mathit{TH}$ in the first $W=\freq{C_1}\cdot N\cdot i$ steps\footnote{We can actually make smaller number of steps, because some steps might have been lost on reaching $C_1$. Nevertheless, overestimating the number of steps is sound.} satisfies $|\maralt{W}(\omega)-\maralt{0}(\omega)|\geq \tau(\omega)\cdot\trend_{C_1}+\frac{5}{6}L_i -2Z\geq \frac{5}{6}L_i-2Z$, where $Z=\max_{s\in\states}|z(s)|$. 
	The number on the right-hand side is positive for all but finitely many $i$. From the Azuma's inequality it follows that the probability of observing such a run is bounded by ${d'}^{(L_i-2Z)^2/W}\leq d^{\sqrt{i}}$ for suitable numbers $d,d'\in(0,1)$ that are independent of $i$.

	For event $F_2$ the argument is similar. 	Note that all runs in $\neg F_1$ make at least $f_{C_1}\cdot N\cdot i - \frac{1}{6}L_i$ steps inside $C_1$, since at most $\frac{1}{6}L_i$ steps were needed to reach $C_1$. If $\omega\in \neg F_1$ increases the counter by at least $\freq{C_1}\cdot N\cdot k\cdot \trend_{C_1} - \frac{1}{6}L_i$ during exactly $W'=f_{C_1}\cdot N\cdot i - \frac{1}{12}L_i$ steps, then it belongs to $\neg F_2$. So assume that $\omega\in \neg F_1$ increases the counter by at most $\freq{C_1}\cdot N\cdot k\cdot \trend_{C_1} - \frac{1}{6}L_i$ during exactly $W'$ steps. Then $|\mart{W'}-\mart{0}(\omega)|\geq \frac{1}{6}(i\cdot N)^{\frac{3}{4}}\cdot\trend_{C_1} -2Z$, where $Z$ is as above. Again, this number is positive for all but finitely many $i$, and for all such $i$ we can apply Azuma's inequality to get that probability of witnessing the small increase is at most $d^{\sqrt{i}}$, where $d$ is a suitable number independent of $k$.
	
	Event $F_5$ is handled in a way which is dual to $F_2 $. We  again use the construction from~\cite{BKK:oc-jacm} to obtain a suitable martingale, which we analyse in almost the same way as in the previous paragraph. The only difference is that since $F_2$ has a negative trend, we now do not bound the probability of a small increase but that of a large decrease.

\subsection{Proofs of Section~\ref{sec-thm2}}
\label{app-thm2}

\begin{reftheorem}{Proposition}{\ref{lem-stable}}
	Suppose that $g^* = 0$, and let $\sigma$ be a strategy which is safe in $s(n)$. Then 
	\[
	\Probm^{\sigma}_{s}(\{\omega \in \run(s) \mid \omega \mbox{ is stable }\}) = 1 \,.
	\] 
\end{reftheorem}

\begin{proof}
We say that a run $\omega = s_0 s_1 \cdots$ in $\calE$ is \emph{drifting} if for every $k \in \Nset$ there exists $i \in \Nset$ such that for all $j \geq i$ we have that $\enlev{0}{}{j} \geq k$. Intuitively, a run is drifting if, for an arbitrary initial counter value, the energy level eventually stays above an arbitrarily large~$k$ along the run.
	
	It follows from the results of \citeaddr{BBEKW:OC-MDP} that the existence of a strategy $\pi$ such that $\pi$ is safe in some configuration $t(m)$ and $\Probm^{\pi}_{t}(\{\omega \in \run(t) \mid \omega \mbox{ is drifting }\}) > 0$ implies the existence of a positive solution of the program $\calT_{\calE}$. 
	
	Suppose that $\sigma$ is a strategy safe in $s(n)$ such that
	\[
	\Probm^{\sigma}_{s}(\{\omega \in \run(s) \mid \omega \mbox{ is stable }\}) < 1 \,.
	\]      
	We show that there exist a configuration $t(m)$ and a strategy $\pi$ with the above properties, and thus derive a contradiction. For every $q \in S$, all $A,B \subseteq S$ where $A \cap B = \emptyset$, and all $f: A \rightarrow \Zset$, let $\run[A_f,B](q)$ be the set of all $\omega \in \run(q)$ such that the set of all control states that appear infinitely often along $\omega$ is precisely $A \cup B$, the set of all control states that are not stable in $\omega$ is precisely $B$, and every control state $r \in A$ is stable at $f(r)$ in $\omega$. Clearly, there must be some $A,f,B$ such that $B \neq \emptyset$ and $\Probm^{\sigma}_{s}(\run[A_f,B](s)) > 0$. For the rest of this proof, we fix such~$A,f,B$. 
	
	For every configuration $r(\ell)$, we define the \emph{$[A_f,B]$-value} of $r(\ell)$ as follows:
	\[
	V_{[A_f,B]}(r(\ell))  :=  \sup\, \{ \Probm^{\varrho}_{r}(\run[A_f,B](r)) \mid \varrho \mbox{ is safe in } r(\ell) \}.
	\]
	Observe that $V_{[A_f,B]}(r(i)) \geq V_{[A_f,B]}(r(j))$ if $i \geq j$. We prove the following:
	\begin{itemize}
		\item[A.] For every $r \in A$, let $r(\ell)$ be the configuration where $\ell = n + f(r)$. Then $V_{[A_f,B]}(r(\ell)) = 1$.
		\item[B.] If $A \neq \emptyset$, then there is a configuration $r(\ell)$ such that $r \in B$ and $V_{[A_f,B]}(r(\ell)) = 1$.
	\end{itemize}
	
	To prove A., let us suppose that there is $r \in A$ such that $V_{[A_f,B]}(r(\ell)) = 1 - \delta$, where $\ell = n + f(r)$ and $\delta > 0$. Let $\omega \in \run[A_f,B](s)$, and consider the sequence of configurations visited by $\omega$ from the initial configuration $s(n)$. Since $r(\ell)$ appears infinitely often in this sequence, we obtain that $\Probm^{\sigma}_{s}(\run[A_f,B](s)) = 0$, which is a contradiction.
	
	To prove B., suppose that there is some $q \in A$, but for all $r \in B$ and $\ell \in \Nset$ we have that $V_{[A_f,B]}(r(\ell)) < 1$. By A., we obtain $V_{[A_f,B]}(q(m)) = 1$ for a suitable~$m$. For every $\omega \in \run[A_f,B](q)$, consider the sequence of configurations visited by $\omega$ from the initial configuration $s(m)$, and let $r(\ell)$ be the first configuration in this sequence such that $r \in B$. Clearly, $\ell \leq m + |S|\cdot \maxE$. Let 
	\[
	V = \max\{V_{[A_f,B]}(u(j)) \mid u \in B, j \leq  m + |S|\cdot \maxE\} \, .
	\]  
	Since $V = 1 - \delta$ for some $\delta > 0$, for every strategy $\varrho$ safe in $s(m)$ we obtain that $\Probm^{\varrho}_{q}(\run[A_f,B](q)) \leq 1 - \delta$, which contradicts $V_{[A_f,B]}(q(m)) = 1$.
	
	The existence of $\pi$ is now proved separately for each of the following two cases:  
	
	\emph{Case I.} Suppose that $V_{[A_f,B]}(r(\ell)) = 1$ for some $r \in B$ and $\ell \in \Nset$. Let us further assume that $\ell$ is the \emph{least} $i$ such that $V_{[A_f,B]}(r(i)) = 1$. A finite path $w$ from $r$ to $r$ of length $j$ is \emph{increasing} if $\enlev{0}{w}{j} > 0$. We claim that for every $\varepsilon > 0$,  
	there exist a strategy $\sigma_\varepsilon$ safe in $r(\ell)$, and $N_\varepsilon \in \Nset$, such that the $\Probm^{\sigma_\varepsilon}_{r}$-probability of all runs initiated in $r$ that start with an increasing path of length at most $N_\varepsilon$ is at least $1 - \varepsilon$. Before proving this claim, let us show how it implies the existence of the promised $t(m)$ and $\pi$. The role of $t(m)$ is taken over by $r(\ell)$. The strategy $\pi$ is constructed as follows.
	Let $\varepsilon_i = 8^{-i}$ for all $i \in \Nset^+$. Consider the strategies $\sigma_{\varepsilon_i}$ and the bounds $N_{\varepsilon_i}$ for all $i \in \Nset^+$. The strategy $\pi$ is defined inductively as follows:
	\begin{itemize}
		\item At the starting state $r$, the strategy $\pi$ ``switches'' to $\sigma_{\varepsilon_1}$.
		\item Whenever $\pi$ ``switches'' to $\sigma_{\varepsilon_j}$, it starts to simulate the strategy $\sigma_{\varepsilon_j}$. If an increasing path is encountered in the first $N_{\varepsilon_j}$ steps from the previous switch, then $\pi$ immediately ``switches'' to $\sigma_{\varepsilon_{j+1}}$. Otherwise, $\pi$ keeps simulating $\sigma_{\varepsilon_j}$ forever. 
	\end{itemize} 
	It follows immediately from the construction of $\pi$ that $\pi$ is safe in $r(\ell)$ and the probability of all runs with infinitely many ``switches'' is at least~$3/4$. Since all runs with infinitely many switches are drifting, we are done.
	
	So, it remains to prove the above claim. Let us fix some $\varepsilon > 0$. Let $\kappa = \varepsilon  \delta /2$, where $\delta$ is either $1$ or $1 - V_{[A_f,B]}(r(\ell{-}1))$, depending on whether $\ell = 0$ or $\ell >0$, respectively (note that $\delta > 0$). We put $\sigma_\varepsilon := \varrho$, where
	$\varrho$ is a strategy safe in $r(\ell)$ such that $\Probm^{\varrho}_{r}(\run[A_f,B](r)) \geq 1 - \kappa$. Note that $\varrho$ is guaranteed to exist, because the \mbox{$[A_f,B]$-value} of $r(\ell)$ is equal to one. Since $r \in B$, for \emph{every} run $\omega = s_0 s_1 s_2 \cdots$ in $\run[A_f,B](r)$ there exist $i < j$ such that $s_i = s_j = r$ and $\enlev{0}{\omega}{i} <  \enlev{0}{\omega}{j}$. We say that $\omega$ is \emph{good} if there are $i < j$ with the above properties such that, in addition, for every $k \leq j$ we have that $s_k = r$ implies $\enlev{0}{\omega}{k} \geq 0$. Now we check that 
	\[
	\Probm^{\varrho}_{r}(\{ \omega \in \run[A_f,B](r) \mid \omega \mbox{ is good }) \geq 1 - \frac{\varepsilon}{2} \,.
	\]
	If $\ell = 0$, the above inequality follows immediately, because then $\varrho$ is safe in $r(0)$. If $\ell > 0$, then the $\Probm^{\varrho}_{r}$ probability of all $\omega \in \run(r)$ that are \emph{not} good runs of $\run[A_f,B](r)$ cannot exceed $\varepsilon/2$, because otherwise, even if all of these runs belong to $\run[A_f,B](r)$, we obtain that $\Probm^{\varrho}_{r}(\run[A_f,B](r))$ is smaller than \[
	(1-\frac{\varepsilon}{2}) + \frac{\varepsilon}{2}(1-\delta) \ = \ 1 - \kappa \, ,
	\]
	which is a contradiction. Since every good run of $\run[A_f,B](r)$ can be recognized after a finite prefix, there must by some $N_\varepsilon$ such that the $\Probm^{\varrho}_{r}$ probability of all good runs of $\run[A_f,B](r)$, where the length this prefix is bounded by $N_\varepsilon$, is at least $1 - \varepsilon$. 
	
	\emph{Case II.} Suppose that $V_{[A_f,B]}(r(\ell)) < 1$ for all $r \in B$ and $\ell \in \Nset$. Note that this implies $A = \emptyset$ by applying claim~B.{} above. For every $\omega \in \run[A_f,B](s)$, let $\alpha_\omega$
	be the sequence of \mbox{$[A_f,B]$-values} of the configurations visited by $\omega$ from the initial configuration $s(n)$. Further, let	$\Lim[A_f,B](\omega) = \liminf_{n \rightarrow \infty} \alpha_\omega$.	We claim that 
	\[
	\Probm^{\sigma}_{s}(\{ \omega \in \run[A_f,B](s) \mid \Lim[A_f,B](\omega) < 1) = 0 \, .
	\]
	Again, let us first show that this claim implies the existence of the promised $t(m)$ and $\pi$. In this case, the role of $t(m)$ is played by $s(n)$, and $\pi$ is chosen as $\sigma$. Since almost all $\omega \in \run[A_f,B](s)$ satisfy $\Lim[A_f,B](\omega) = 1$, it suffices to show that every run $\omega = s_0 s_1\cdots$ of $\run[A_f,B](s)$ such that $\Lim[A_f,B](\omega) = 1$ is drifting. However, since $V_{[A_f,B]}(r(\ell)) < 1$ for all $r \in B$ and $\ell \in \Nset$, it follows immediately that for all $r \in B$ and  $k \in \Nset$ there exists $i \in \Nset$ such that for all $j \geq i$ we have that $s_j = r$ implies $\enlev{0}{\omega}{j} \geq k$. So, $\omega$ is indeed drifting.
	
	It remains to prove the above claim. It suffices to show that for every fixed $\varepsilon > 0$ we have that
	\[
	\Probm^{\sigma}_{s}(\{ \omega \in \run[A_f,B](s) \mid \Lim[A_f,B](\omega) < 1 - \varepsilon) = 0 \, .
	\]
	Let $\omega \in \run[A_f,B](s)$ be a run such that $\Lim[A_f,B](\omega) < 1 - \varepsilon$, and let us consider the sequence of configurations visited by $\omega$ from the initial configuration $s(n)$. Clearly, this sequence visits infinitely often a configuration whose \mbox{$[A_f,B]$-value} is bounded by $1 - \varepsilon$, which implies that the total probability of all such runs is zero.   
\qed
\end{proof}

\subsection{A Proof of Theorem~\ref{thm:onedim-scc-pump} (5.)}
As explained in Section~\ref{sec-intro}, the problem whether a given configuration of EMDP is safe is equivalent to solving the corresponding energy game (with the same transition structure as the EMDP). To finish the proof of Theorem~\ref{thm:onedim-scc-pump} (5.), we need to show that it suffices to restrict to pumpable EMDPs. 

So let us fix an EMDP $\calE=(\M,E)$ where $\M = (\states,(\cstates,\stochstates),\trans{},\Prob,\rev)$.
We define an EMDP $\calE'=(\M',E')$ where the set of states is $S\cup T$, from each $s\in S$ there are transitions to all elements of $\out(s)$, from each $(s,s')\in T$ there are transitions to $(s,s')$ and to $s'$. The set of stochastic states of $\M'$ is $\stochstates\cup T$. The probability of each transition $(s,(s,s'))$, here $s\in \stochstates$, in $\M'$ is equal to the probability of $(s,s')$ in $\M$. The probability of each transition $((s,s'),(s,s'))$ in $\M'$ is equal to $\frac{1}{2}$.
The energy update function $E'$ is defined by $E'(s,(s,s'))=E(s,s')$ and $E'((s,s'),(s,s'))=\max_{e\in T} E(e)+1$ and $E'((s,s'),s')=0$. The reward function in $\M'$ can be defined arbitrarily (we are concerned only with safety). 

Now note that a configuration $s(n)$ is safe in $\M$ iff $s(n)$ is safe in $\M'$. So $\val(s(n))>-\infty$ in $\M'$ iff $\val(s(n))>-\infty$ in $\M$ iff $s(n)$ is safe in the corresponding energy game on $\M$.
 Also, note that $\M'$ is pumpable since in every $(s,s')$ the counter may be pumped above any bound with a positive probability, which eventually happens with probability one.